\documentclass[journal,12pt,onecolumn,draftclsnofoot,]{IEEEtran}

\usepackage[latin9]{inputenc}
\usepackage{textcomp}
\usepackage{amsmath}
\usepackage{amssymb}
\usepackage{tabularx,ragged2e,booktabs}
\usepackage{epsfig,rotating,setspace,latexsym,amsmath,epsf,amssymb,bm}
\usepackage{cite,graphicx,color,subfigure, amsthm}
\usepackage{subfigure}
\usepackage{multicol}
\usepackage{thmtools}
\usepackage{etoolbox}
\usepackage{enumitem}
\usepackage[autostyle]{csquotes}

\newtheorem{prop}{Proposition}

\newtheorem{lemma}{Lemma}

\usepackage{flushend}
\usepackage{lipsum}

\usepackage[english]{babel}

\ifCLASSINFOpdf

\else

\fi

\begin{document}

\title{Online Edge Caching and Wireless Delivery in Fog-Aided Networks with Dynamic Content Popularity}

\author{Seyyed Mohammadreza~Azimi, 
        Osvaldo~Simeone,~
         Avik~Sengupta
        ~and ~Ravi~Tandon
\thanks{S. M. Azimi is with the CWiP, Department
of Electrical and Computer Engineering, New Jersey Institute of Technology, Newark,
NJ, USA. E-mail: (sa677@njit.edu). O. Simeone is with Centre for Telecommunications Research
Department of Informatics, King's College London, UK. Email: (osvaldo.simeone@kcl.ac.uk) 
A. Sengupta is with Next Generation and Standards, iCDG, Intel Corporation, Santa Clara, USA. Email: (avik.sengupta@intel.com). R. Tandon is with University of Arizona, Tucson, AZ, USA. E-mail: (tandonr@email.arizona.edu). The work of S.M. Azimi and O. Simeone was partially supported by the U.S. NSF through grant CCF-1525629. }%
}

\markboth{Draft}%
{Shell \MakeLowercase{\textit{et al.}}: Bare Demo of IEEEtran.cls for IEEE Journals}

\maketitle

\begin{abstract}
Fog Radio Access Network (F-RAN) architectures can leverage both cloud processing and edge caching for content delivery to the users. To this end, F-RAN utilizes caches at the edge nodes (ENs) and fronthaul links connecting a cloud processor to ENs. Assuming time-invariant content popularity, existing information-theoretic analyses of content delivery in F-RANs rely on offline caching with separate content placement and delivery phases. In contrast, this work focuses on the scenario in which the set of popular content is time-varying, hence necessitating the online replenishment of the ENs' caches along with the delivery of the requested files.
The analysis is centered on the characterization of the long-term Normalized Delivery Time (NDT), which captures the temporal dependence of the coding latencies accrued across multiple time slots in the high signal-to-noise ratio regime. Online edge caching and delivery schemes are investigated for both serial and pipelined transmission modes across fronthaul and edge segments. Analytical results demonstrate that, in the presence of a time-varying content popularity, the rate of fronthaul links sets a fundamental limit to the long-term NDT of F-RAN system. Analytical results  are further verified by numerical simulation, yielding important design insights.
\end{abstract}


\IEEEpeerreviewmaketitle

\section{Introduction}
Delivery of wireless multimedia content poses one of the main challenges in the definition of enhanced mobile broadband services in $5$G (see, e.g., \cite{Wong}).
 In-network caching, including \textit{edge caching}, is a key technology for the deployment of information-centric networking with reduced bandwidth \cite{Ali5}, latency \cite{OS}, and energy \cite{Jaime}. Specifically, edge caching stores popular content at the edge nodes (ENs) of a wireless system, thereby reducing latency and backhaul usage when the requested contents are cached \cite{Shan}. 
 
While edge caching moves network intelligence closer to the end users, \textit{Cloud Radio Access Network} (C-RAN) leverages computing resources at a central \enquote{cloud} processing unit, and infrastructure communication resources in the form of fronthaul links connecting cloud to ENs \cite{Sim1, Quek}. The cloud processor is typically part of the transport network that extends to the ENs, and is also known as "edge cloud" \cite{Gpp}. The C-RAN architecture can be used for content delivery as long as the cloud has access to the content library \cite{OS, Avik2}.  

Through the use of Network Function Virtualization (NFV) \cite{Her}, 5G networks will enable network functions to be flexibly allocated between edge and cloud elements, hence breaking away from the purely edge- and cloud-based solutions provided by edge caching and C-RAN, respectively. To study the optimal operation of networks that allow for both edge and cloud processing, references \cite{OS, Avik2} investigated a \textit{Fog Radio Access Network (F-RAN)} architecture, in which the ENs are equipped with limited-capacity caches and fronthaul links. These works addressed the optimal use of the communication resources on the wireless channel and fronthaul links and storage resources at the ENs, under the assumption of offline caching. With offline caching, caches are replenished periodically, say every night, and the cached content is kept fixed for a relatively long period of time, e.g., throughout the day, during which the set of popular contents is also assumed to be invariant.

\emph{Related Work}: 
The information theroretic analysis of offline edge caching in the presence of a static set of popular contents was first considered in \cite{Ali2}. In this seminal work, an achievable number of degrees of freedom (DoF), or more precisely its inverse, is determined as a function of cache storage capacity for a system with three ENs and three users. In \cite{Tao, Navid, Salman, jad, Roig}, generalization to the scenario with cache-aided transmitters as well as receivers is considered. In particular, in \cite{Navid, Salman}, it is proved that, under the assumption of one-shot linear precoding, the maximum achievable sum-DoF of a wireless system with cache-aided transmitters and receivers scales linearly with the aggregate cache capacity across the nodes with both fully connected and partially connected topologies. In \cite{jad}, separation of network and physical layers is proposed, which is proved to be approximately optimal. Reference \cite{Roig} extended the works in \cite{Tao, Navid, Salman, jad} to include decentralized content placement at receivers. 

In contrast to abovementioned works, references \cite{Avik1, Avik2, Koh, jasper,Kakar}  investigated the full F-RAN scenario with both edge caching and cloud processing in the presence of offline caching. Specifically, reference \cite{Avik2} derives upper and lower bounds on a high-SNR coding latency metric defined as Normalized Delivery Time (NDT), which generalizes the inverse-of-DoF metric studied in \cite{Ali2} and in most of the works reviewed above. 
The minimum NDT is characterized within a multiplicative factor of two. While in \cite{Avik2} it is assumed that there are point-to-point fronthaul links with dedicated capacities between cloud and ENs, a wireless broadcast channel is considered between cloud and ENs in \cite{Koh}. The scenario with heterogeneous cache requirements is considered in \cite{jasper}.  
Offline caching for a small-cell system with limited capacity fronthaul connection between small-cell BS and cloud-processor and partial wireless connectivity is invistigated in \cite{Kakar} and \cite{swiss}  under different channel models.  
\begin{figure}\label{ref_off}
\centering
\includegraphics[width=.8\textwidth]{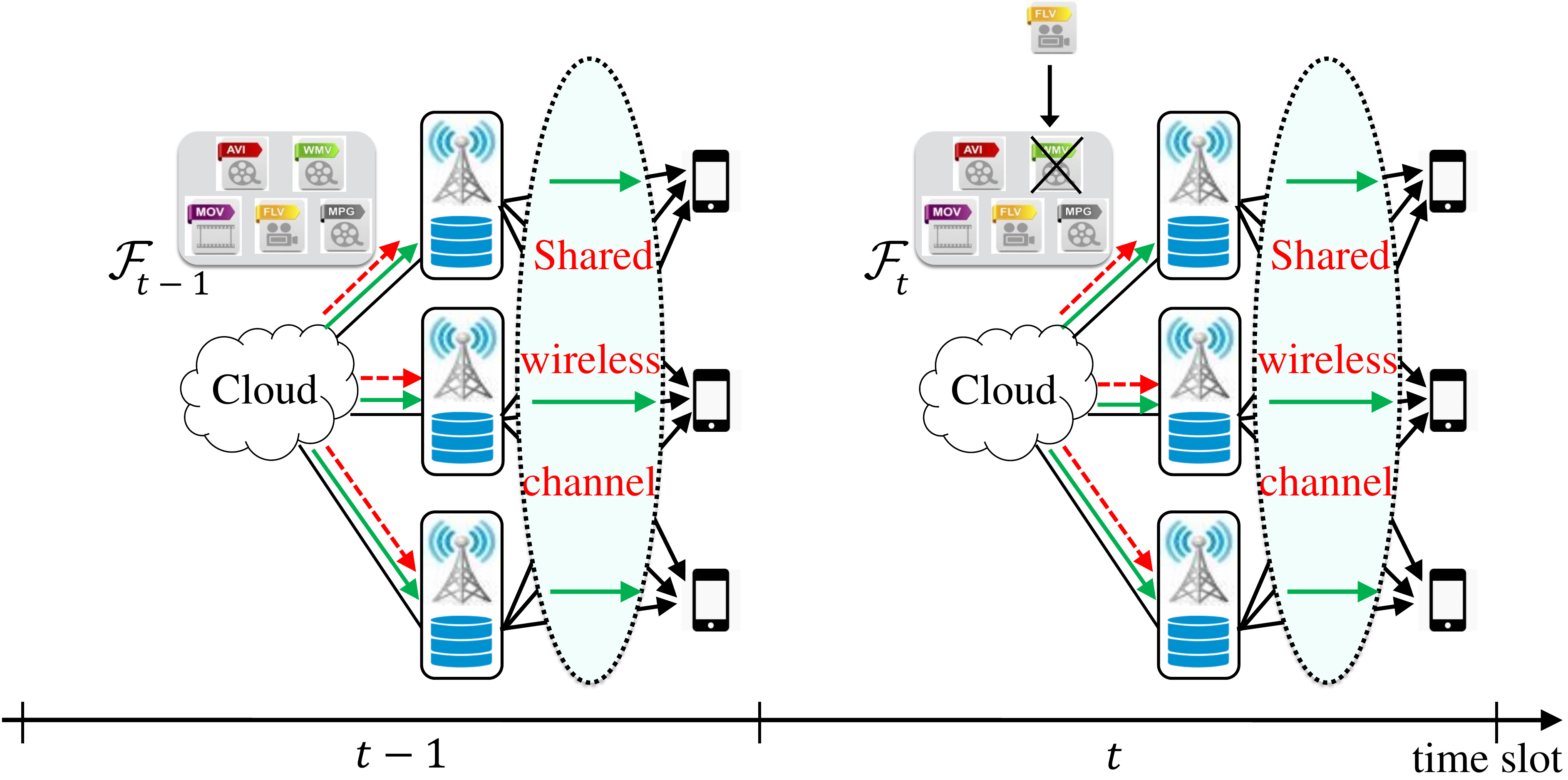}
\caption{With online edge caching, the set of popular files $\mathcal{F}_t$ is time-varying, and online cache update (red dashed arrows) generally can be done at each time slot along with content delivery (green arrows).}
\label{oncac}
\end{figure}

\emph{Main Contributions}: 
In this work, we consider an alternative \emph{online} caching set-up, typically considered in the networking literature \cite{Leo1}, in which the set of popular files is time-varying, making it generally necessary to perform cache replenishment as the system delivers contents to the end users. With the exception of conference version of this manuscript \cite{Azimi}, to the best of authors' knowledge, this is the first work to present an information theoretic analyses of online caching for edge caching and F-RANs. The main contributions are as follows.     

 \noindent $\bullet$ The performance metric of the long-term NDT, which captures the temporal dependence of the high-SNR coding latencies accrued in different slots, is introduced in Sec. \ref{proform};

 \noindent $\bullet$ Online edge caching and delivery schemes based on reactive online caching are proposed for a set-up with serial fronthaul-edge transmission, and bounds on the corresponding achievable long-term NDTs are derived in Sec. \ref{Fundamental limits} by considering both fixed and adaptive caching;

\noindent $\bullet$ Reactive and proactive online caching schemes are proposed for a pipelined fronthaul-edge transmission mode, and bounds on the corresponding achievable long-term NDTs are derived in Sec. \ref{sec_para};

 \noindent $\bullet$ The performance loss caused by the time-varying content popularity in terms of delivery latency is quantified by comparing the NDTs achievable under offline and online edge caching in Sec. \ref{Comp};

 \noindent $\bullet$ Numerical results are provided in Sec. \ref{Numer} that offer insights into the comparsion of  reactive online edge caching schemes with different eviction mechanisms, such as random, Least Recently Used (LRU) and First In First Out (FIFO) (see, e.g., \cite{Leo1}), proactive online edge caching schemes, under both serial and pipelined transmission.

The material in this paper was partly presented in the conference paper \cite{Azimi}. The additional contributions of this journal version include the analysis of reactive online caching with known popular set for serial fronthaul-edge transmission; the design of adaptive online edge caching as a means to optimize the performance of online caching; the study of pipelined fronthaul-edge transmission; and the inclusion of all proofs as well as new numerical results\footnote{We also correct an erroneous conclusion about proactive caching in \cite{Azimi} in Sec. \ref{sec_para}.}. 

\textit{Notation}: Given random variable $X$, the corresponding entropy is denoted by $\textrm{H}(X)$. The equality $f(x)=O(g(x))$ indicates the relationship $\lim_{x \rightarrow \infty} |f(x)/g(x)| < \infty$.

\section{System Model and Performance Criterion}\label{proform}
In this section, we first present the system model adopted for the analysis of F-RAN systems with online edge caching.
Then, we introduce the long-term NDT as the performance measure of interest. We focus here on the scenario with serial fronthaul-edge transmission and treat the pipelined case in Sec. \ref{sec_para}. As detailed below, with serial delivery, fronthaul transmission is followed by edge transmission, while pipelining enables simultaneous transmission on fronthaul and edge channels.
\subsection{System Model}\label{sysmo1}
We consider the $M \times K$ F-RAN with online edge caching shown in Fig. \ref{oncac}, in which $M$ ENs serve a number of users through a
shared downlink wireless channel. Time is organized into time slots, and $K$ users can request contents in each time slot. The requested contents need to be delivered within the given time slot, which may generally accommodate also other transmissions not explicitly modeled here. Each EN is connected to the cloud via a dedicated fronthaul link with capacity $C_F$ bits per symbol period, where the latter henceforth refers to the symbol period for the wireless channel. Each time slot contains a number of symbols. In each time slot $t$, $K$ users make requests from the current set $\mathcal{F}_t$ of $N$ popular files. All files in the popular set are assumed to have the same size of $L$ bits. As per standard modelling assumptions (see, e.g., \cite{Ali5}), each file can be interpreted as a chunk of a video, or of some other multimedia content, that needs to be delivered within a time slot. The cache capacity of each EN is $\mu NL$ bits, where $\mu$, with $0 \leq \mu \leq 1$, is defined as the \emph{fractional cache capacity} since $NL$ is the dimension of set $\mathcal{F}_t$ in bits.

\textit{Time-varying popular set}: At each time slot $t$, each of the $K$ active users requests a file from the time-varying set $\mathcal{F}_t$ of $N$ popular files, with $t \in \{1,2,...\}.$ The indices of the files requested by the $K$ users are denoted by the vector $d_t=(d_{1,t},...,d_{K,t})$, where $d_{k,t}$ represents the file requested by user $k$. As in \cite{Ramtin}, indices are chosen uniformly without replacement in the interval $[1:N]$ following an arbitrary order. 

The set of popular files $\mathcal{F}_t$ evolves according to the Markov model considered in \cite{Ramtin}. The model posits that, at the time scale of the delivery of the chunks in the library, say a few milliseconds, the set of popular files is slowly time-varying. As such, at this scale, a Markovian model can closely approximate more complex memory models. Reference \cite{Ramtin} performed an experimental study on the movie rating of films released after 2005 by Netflix and found the assumption to be well suited to represent real popularity dynamics.  Markov models are also routinely used in the networking literature, see, e.g., \cite{Ps}. Accordingly, given the popular set $\mathcal{F}_{t-1}$ at time slot $t-1$, with probability $1-p$, no new popular content is generated and we have $\mathcal{F}_t=\mathcal{F}_{t-1}$; while, with probability $p$, a new popular file is added to the set $\mathcal{F}_t$ by replacing a file selected uniformly at random from $\mathcal{F}_{t-1}$. We consider two cases, namely: $(i)$ \textit{known popular set}:  the cloud is informed about the set $\mathcal{F}_t$ at time $t$, e.g., by leveraging data analytics tools; $(ii)$ \textit{unknown popular set}: the set $\mathcal{F}_t$ may only be inferred via the observation of the users' requests. We note that the latter assumption is typically made in the networking literature (see, e.g., \cite{Leo1}). Furthermore, the known popular set case can reflect scenarios in which the offering of the content provider changes over time in a controlled way, e.g., for curated content services. It also provides a useful baseline scenario for the practical and complex case with unknown popular set.

\textit{Edge channel}: The signal received by the $k$th user in any symbol of the time slot $t$ is
\begin{equation}\label{rcvr}
Y_{k,t}=\sum_{m=1}^{M}H_{k,m,t}X_{m,t} + Z_{k,t},
\end{equation}
where $H_{k,m,t}$ is the channel gain between $m$th EN and $k$th user at time slot $t$; $X_{m,t}$ is the signal transmitted by the $m$th EN; and $Z_{k,t}$ is additive noise at $k$th user. The channel coefficients are assumed to be independent and identically distributed (i.i.d.) according to a continuous distribution and to be time-invariant within each slot. Also, the additive noise $Z_{k,t} \sim \mathcal{CN}(0,1)$ is i.i.d. across time and users.  At each time slot $t$, all the ENs, cloud and users have access to the global channel state information (CSI) about the wireless channels $H_t=\{\{H_{k,m,t} \}_{k=1}^{K} \}_{m=1}^{M}$.

\begin{figure}[t]
\vspace{-5mm}
\centering
\includegraphics[width=.6\textwidth]{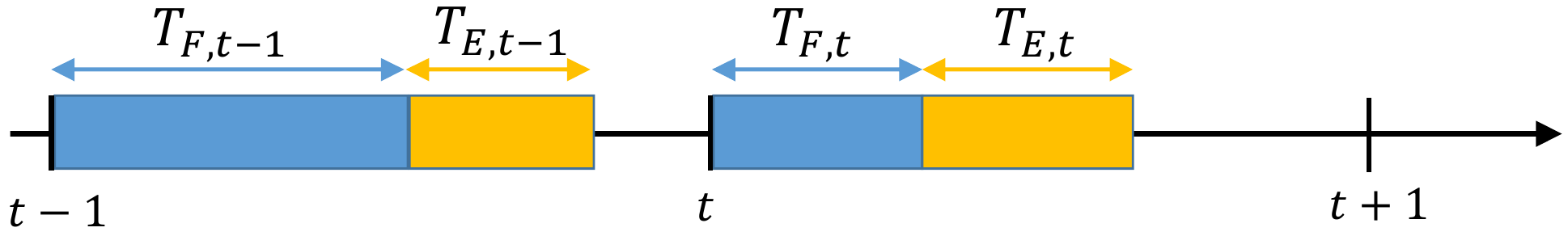}
\caption{ Illustration of the definition of fronthaul and edge delivery times for the serial deliver mode as a function of the time slot index $t$.}
\label{TS_fig}
\vspace{-6mm}
\end{figure}

\textit{System operation}: The system operates according to a combined fronthaul, caching, edge transmission and decoding policy. As illustrated in Fig. \ref{TS_fig}, in the delivery phase, at each time slot $t$, a number of symbol durations $T_{F,t}$ are allocated to fronthaul transmission and a number $T_{E,t}$ to transmission on the edge channel. In general, these durations are a function of the time slot index $t$, since the durations depend on the randomly selected request vector $d_t$ and on the currently cached files. Note that the duration of a time slot is immaterial to the analysis, as it is assumed to be sufficiently large to accommodate the overall transmission time $T_t= T_{F,t}+ T_{E,t}$. Formal definitions are provided below. 

 \noindent $\bullet$  \textit{Fronthaul policy}: The cloud transmits a message $U_{m,t}$ to each EN $m$ in any time slot $t$ as a function of the current demand vector $d_{t}$,  ENs' cache contents, and CSI $H_{t}$, as well as, in the case of known popular set, the set $\mathcal{F}_t$. The fronthaul capacity limitations impose the condition $\textrm{H}(U_{m,t})\leq T_{F,t}C_F$, where $T_{F,t}$ is the duration (in symbols) of the fronthaul transmission $U_{m,t}$ in time slot $t$ for all ENs $m=1,...,M$. 

 \noindent $\bullet$ \textit{Caching policy}:  After fronthaul transmission, in each time slot $t$, any EN $m$ updates its cached content $S_{m,t-1}$ in the previous slot based on the fronthaul message $U_{m,t}$, producing the updated cache content $S_{m,t}$. Due to cache capacity constraints, we have the inequality $\textrm{H}(S_{m,t}) \leq \mu NL$, for all slots $t$ and ENs $m$. More specifically, as in \cite{Avik2}, we allow only for intra-file coding. Therefore, the cache content $S_{m,t}$ can be partitioned into independent subcontents $S^f_{m,t}$, each obtained as a function of a single file $f$. In the case of known popular set, we can assume without loss of generality that only popular files are cached. Furthermore, we constrain all files to be cached with the same number of bits at all ENs, yielding the constraint $ H(S_{m,t}^ f) \leq \mu L$. As shown in \cite[Sec. III-A ]{Avik2}, this is known to be approximately optimal in the offline set-up. Under an unknown set of popular files, we cannot generally guarantee that only popular files are cached. Consistently with the choices made for the case of known popular set, we assume that, at each time $t$, a certain number $N_c$ of files is cached at all ENs and that the stored fraction for each cached file $f$ must satisfy the constraint $H(S_{m,t}^ f)<=\mu(N/Nc)L$.

 \noindent $\bullet$ \textit{Edge transmission policy}: Upon updating the caches, the edge transmission policy at each EN $m$ transmits the codeword $X_{m,t}$, of duration $T_{E,t}$ on the wireless channel as a function of the current demand vector $d_t$, CSI $H_t$, cache contents $S_{m,t}$ and fronthaul messages $U_{m,t}$. We assume a per-slot power constraint $P$ for each EN. 

 \noindent $\bullet$ \textit{Decoding policy}: Each user $k$ maps its received signal $Y_{k,t}$ in \eqref{rcvr} over a number $T_{E,t}$ of channel uses to an estimate $\hat{F}_{d_{t,k}}$ of the demanded file $F_{d_{t,k}}$.

Denoting a joint fronthaul/caching/edge transmission and decoding policy as $\Pi$, the probability of error of a policy $\Pi$ at slot $t$ is defined as the worst-case probability
\begin{equation}
\begin{aligned}\label{error}
\begin{array}{ll}
\textrm{P}_{e,t}=\underset{k \in \{1,...,K\}}{\max}\textrm{Pr}(\hat{F}_{d_{t,k}} \neq F_{d_{t,k}}),
\end{array}
\end{aligned}
\end{equation}
which is evaluated over the distributions of the popular set $\mathcal{F}_t$, of the request vector $d_t$ and of the CSI $H_t$. A sequence of policies $\Pi$ indexed by the file size $L$ is said to be \textit{feasible} if, for all $t$, we have $\textrm{P}_{e,t} \rightarrow 0$ when $L \rightarrow \infty$.

\subsection{Long-term Normalized Delivery Time (NDT)}\label{NDT_S}
For given parameters $(M,K,N,\mu,C_F,P)$, the average delivery time per bit in slot $t$ achieved by a feasible policy under serial fronthaul-edge transmission is defined as the sum of fronthaul and edge contributions 
\begin{equation}\label{DTB_S}
\Delta_t(\mu,C_F,P)=\Delta_{F,t}(\mu,C_F,P) + \Delta_{E,t}(\mu,C_F,P), \\
\end{equation}
where the fronthaul and edge latencies per bit are given as 
\begin{align}\label{DTB_S_F}
\Delta_{F,t}(\mu,C_F,P)&= \underset{L \rightarrow \infty}{\lim}\frac{1}{L}\textrm{E}  [T_{F,t}]  ~~\text{and}~~ 
\Delta_{E,t}(\mu,C_F,P)= \underset{L \rightarrow \infty}{\lim}\frac{1}{L}\textrm{E}  [T_{E,t}].
\end{align}
In \eqref{DTB_S_F}, the average is taken with respect to the distributions of $\mathcal{F}_t$, $d_{t}$ and $H_{t}$, and we have made explicit the dependence only on the system resource parameters $(\mu,C_F,P)$.

As in \cite{Avik2}, in order to evaluate the impact of a finite fronthaul capacity in the high-SNR regime, we let the fronthaul capacity scale with the SNR parameter $P$ as $C_F=r \log (P)$, where $r \geq 0$ measures the ratio between fronthaul and wireless capacities at high SNR.
Furthermore, we study the scaling of the latency with respect to a reference system in which each user can be served with no interference at the Shannon capacity $\log(P) + o(\log P)$. Accordingly, for any achievable delivery time per bit \eqref{DTB_S_F}, the Normalized Delivery Times for fronthaul and edge transmissions in time slot $t$  \cite{Avik2} are defined as   
\begin{align}\label{FNDT}
\delta_{F,t}(\mu,r) &=\underset{P \rightarrow \infty}{\lim}\frac{\Delta_{F,t}(\mu,r \log (P),P)}{1/\log(P)} ~~\text{and}~~  
\delta_{E,t}(\mu,r) =\underset{P \rightarrow \infty}{\lim}\frac{\Delta_{E,t}(\mu,r \log (P),P)}{1/\log(P)},
\end{align}
respectively. 
In \eqref{FNDT},  the delivery time(s) per bit in \eqref{DTB_S_F} are normalized by the term $1 / \log(P)$, which measures the delivery time per bit at high SNR of  the mentioned reference system \cite{Avik2}. The NDT in time slot $t$ is defined as the sum $\delta_t(\mu,r)=\delta_{E,t}(\mu,r)+\delta_{F,t}(\mu,r)$.

In order to capture the memory entailed by online edge caching policies on the system performance, we introduce the \textit{long-term NDT} metric. To this end, we take the standard approach considered in dynamic optimization problems of evaluating the performance in terms of averages over a long time horizon (see, e.g., \cite{MDP}).
\begin{equation}\label{avNDT}
\bar{\delta}_{\text{on}}(\mu,r)=\underset{T \rightarrow \infty}{\lim \sup} \frac{1}{T}\underset{t=1}{\overset{T}{\sum}} \delta_t(\mu,r).
\end{equation}
This metric provides a measure of the expected number of resources needed for content delivery and on the effect of caching decisions on the long-term performance of the system.
We denote the minimum long-term NDT over all feasible policies under the known popular set assumption as $\bar{\delta}^*_{\text{on,k}}(\mu,r)$, while $\bar{\delta}^*_{\text{on,u}}(\mu,r)$ denotes the minimum long-term NDT under the unknown popular set assumption.

As a benchmark, we also consider the minimum NDT
for offline edge caching $\delta_{\text{off}}^*(\mu,r)$ as studied in \cite{Avik2}. By definition, we have the inequalities $\delta_{\text{off}}^*(\mu,r) \leq \bar{\delta}_{\text{on,k}}^*(\mu,r) \leq \bar{\delta}_{\text{on,u}}^*(\mu,r)$. Note that the first inequality merely indicates that the performance in the presence of a static set of popular files can be no worse than in the presence of a time-varying set of popular files.

\section{Preliminaries: Offline caching}\label{OFF}
In this section, we first summarize for reference some key results on offline caching in F-RAN from \cite{Avik2}. With offline caching, the set of popular files $\mathcal{F}_t=\mathcal{F}$ is time invariant and caching takes place in a separate placement phase. Reference \cite{Avik2} identified offline caching and delivery policies that are optimal within a multiplicative factor of $2$ in terms of NDT achieved in each time slot. The policies are based on fractional caching, whereby an uncoded fraction $\mu$ of each file is cached at the ENs, and on three different delivery approaches, namely EN cooperation, EN coordination, and C-RAN transmission.
 
\textit{EN cooperation}: This approach is used when all ENs cache all files in the library. In this case, joint Zero-Forcing (ZF) precoding can be carried out at the ENs so as to null interference at the users. This can be shown to require an edge and fronthaul-NDTs in \eqref{FNDT} equal to \cite{Avik2}
\begin{equation}\label{coop_e}
\delta_{\text{E,Coop}}=\frac{K}{\min\{M,K\}} ~\text{and}~\delta_{\text{F,Coop}}=0
\end{equation}
in order to communicate reliably the requested files to all users. Note that, when the number of ENs is larger than the number of active users, i.e., $M \geq K$, we have $\delta_{\text{E,Coop}}=1$, since the performance becomes equivalent to that of the considered interference-free reference system (see Sec. \ref{NDT_S}). For reference, we also write the fronthaul-NDT $\delta_{\text{F,Coop}}=0$ since this scheme does not use fronthaul resources.

\textit{EN coordination}: This approach is instead possible when the ENs store non-overlapping fractions of the requested files. Specifically, if each EN caches a different fraction of the popular files, \textit{Interference Alignment (IA)} can be used on the resulting so-called X-channel\footnote{In an X-channel, each transmitter has an independent message for each receiver \cite{Jafar}.}. This yields the edge and fronthaul-NDTs \cite{Avik2}
\begin{equation}\label{coor_e}
\delta_{\text{E,Coor}}=\frac{M+K-1}{M}, ~\text{and}~\delta_{\text{F,Coor}}=0.
\end{equation}

\textit{C-RAN transmission}: While EN coordination and cooperation are solely based on edge caching, C-RAN transmission uses cloud and fronthaul resources.  Specially, C-RAN performs ZF precoding at the cloud, quantizes the resulting signals and sends them on the fronthaul links to the ENs. The ENs act as relays that transmit  the received fronthaul messages on the wireless channel. The resulting edge and fronthaul-NDTs are equal to  \cite{Avik2}
\begin{equation}\label{cran_e}
\delta_{\text{E,C-RAN}}(r)=\delta_{\text{E,Coop}}=\frac{K}{\min\{M,K\}} ~\text{and}~\delta_{\text{F,C-RAN}}(r)=\frac{K}{Mr}.
\end{equation}
Note that the edge-NDT is the same as for EN cooperation due to ZF precoding at the cloud, while the fronthaul NDT is inversely proportional to the fronthaul rate $r$. 
 
\textbf{Offline caching policy}: 
In the placement phase, the offline caching strategy operates differently depending on the values of the fronthaul rate $r$. 

\noindent $\bullet$ \emph{Low fronthaul regime}: If the fronthaul rate $r$ is smaller than a threshold  $r_{th}$, the scheme attempts to maximize the use of the ENs' caches by distributing the maximum fraction of each popular file among all the ENs. When $\mu \leq 1/M $, this is done by storing  non-overlapping fractions of each popular file  at different ENs, leaving a fraction uncached (see top-left part of Fig. \ref{seroff}). When $\mu \geq 1/M $, this approach yields a fraction $(\mu M-1)/(M-1)$ of each file that is shared by all ENs with no uncached parts (see top-right part of Fig. \ref{seroff}). The threshold is identified in \cite{Avik2} as $r_{th}=K(M-1)/(M(\min\{M,K\}-1))$.  

\noindent $\bullet$ \emph{High fronthaul}: If $r \geq r_{th}$, a common $\mu$-fraction of each file is placed at all ENs, as illustrated in the bottom part of Fig. \ref{seroff}, in order to maximize the opportunities for EN cooperation, hence always leaving a fraction uncached unless $\mu=1$. 

\textbf{Offline delivery policy}: In the delivery phase, the policy operates as follows. With reference to Fig. \ref{seroff}, fractions of each requested file stored at different ENs are delivered using EN coordination; the uncached fractions are delivered using C-RAN transmission; and fractions shared by all ENs are delivered using EN cooperation. Time sharing between pairs of such strategies is used in order to transmit different fractions of the requested files. For instance, when $r \geq r_{th}$, the approach time-shares between EN cooperation and C-RAN transmission (see bottom of Fig. \ref{seroff}).

\textbf{Achievable NDT:} We denote the achievable NDT in each time slot of the outlined offline caching and delivery policy by $\delta_{\text{off,ach}}(\mu,r)$. Analytical expression for $\delta_{\text{off,ach}}(\mu,r)$ can be obtained from \eqref{coop_e}-\eqref{cran_e} using time sharing. In particular if a fraction $\lambda$ of the requested file is delivered using a policy with NDT $\delta'$ and the remaining fraction using a policy with NDT $\delta''$, the overall NDT is given by $\lambda \delta'+(1-\lambda) \delta''$ (see \cite[Proposition 4]{Avik2} for more details). Letting $\delta^*_{\text{off}}(\mu,r)$ be the minimum offline NDT, the achievable NDT of the offline caching and delivery policy was proved to be within a factor of $2$ of optimality in the sense that we have the inequality \cite[Proposition 8]{Avik2}
\begin{equation}
\begin{aligned}\label{achtomin}
\frac{\delta_{\text{off,ach}}(\mu,r)}{\delta^*_{\text{off}}(\mu,r)} \leq 2.
\end{aligned}
\end{equation}

\begin{figure}[t]\label{ref_off}
\vspace{-5mm}
\centering
\includegraphics[width=.6\textwidth]{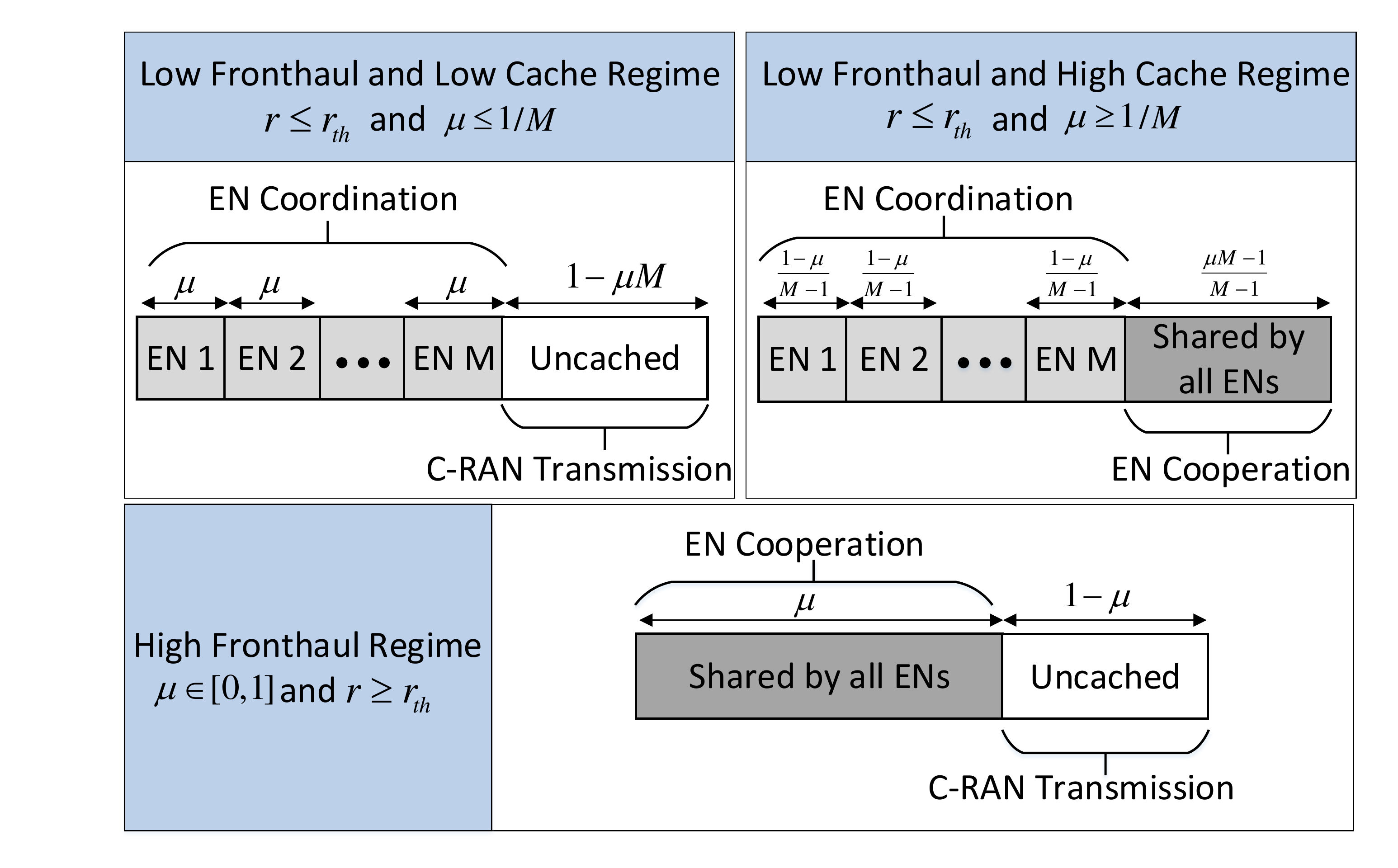}
\caption{ Illustration of the offline caching policy proposed in \cite{Avik2}. Note that each EN caches a fraction $\mu$ of each file.}
\label{seroff}
\vspace{-6mm}
\end{figure}

\section{Achievable Long-Term NDT}\label{Fundamental limits}
In this section, we propose online edge caching with fronthaul-edge transmission policies operating under known and unknown popular set assumptions, and evaluate their performance for serial fronthaul-edge transmission. Lower bounds on the minimum long-term NDT will be presented in Sec. \ref{sec_para}.

\subsection{C-RAN Delivery}\label{Griddy}
C-RAN delivery neglects the cached contents and uses C-RAN transmission in each time slot. This achieves a long-term NDT that coincides  with the offline NDT \eqref{cran_e} obtained in each time slot, i.e.,
\begin{equation}\label{greedy}
\bar{\delta}_{\text{C-RAN}}(r)=\delta_{\text{E,C-RAN}}(r)+\delta_{\text{F,C-RAN}}(r)=\frac{K}{\min\{M,K\}}+\frac{K}{Mr}.
\end{equation}

\subsection{Reactive Online Caching with Known Popular Set}\label{ReactK_S}
We now consider online caching schemes by focusing on reactive strategies that update the ENs' caches every time an uncached file is requested by any user. Discussion on proactive strategies is postponed to Sec. \ref{sec_para}.  We start by considering the simple case of known popular set, and study the unknown popular set case in the next subsection. 

If, at time slot $t$, $R_t$ requested files, with $0 \leq  R_t \leq K$, are not cached at the ENs, a fraction of each requested and uncached file is reactively sent on the fronthaul link to each EN at the beginning of the time slot. To select this fraction, we follow the offline caching policy summarized in Sec. \ref{OFF} and Fig. \ref{seroff}. Therefore, we cache a fraction $\mu$ of the file at all ENs, where the fraction is selected depending on the values of $r$ and $\mu$ as in Fig. \ref{seroff}.  An improved selection of the size of cached fraction will be  discussed later in this section.  In order to make space for new files, the ENs evict files that are no longer popular as instructed by the cloud. Note that this eviction mechanism is feasible since the cloud knows the set $\mathcal{F}_t$. Furthermore, it is guaranteed to satisfy the ENs' cache capacity of $\mu N L$ bits in each time slot. As a result of the cache update, each requested file is cached as required by the offline delivery strategy summarized in Sec. \ref{OFF} and Fig. \ref{seroff}, which is adopted for delivery. 

The overall NDT is hence the sum of the NDT $\delta_{\text{off,ach}}(\mu,r)$ achievable by the offline delivery policy described in Sec. \ref{OFF} and of the NDT due to the  fronthaul transfer of the $\mu$-fraction of each requested and uncached file on the fronthaul link. By \eqref{FNDT}, the latter equals $(\mu/C_F) \times \log P=\mu / r$, and hence the  achievable NDT at each time slot $t$ is
\begin{equation}\label{k_offcent_time}
\delta_{t} (\mu,r)=\delta_{\text{off,ach}} (\mu,r )+\frac{\mu \textrm{E}[R_t]}{r}.
\end{equation}
The following proposition presents the resulting achievable long-term NDT of reactive online caching with known popular set.
\begin{prop}\label{known_lem}
For an $M \times K$ F-RAN with $N \geq K $, in the known popular set case, online reactive caching achieves the long-term NDT
\begin{equation}
\begin{aligned}\label{react_kkk}
 \bar{\delta}_{\emph{react,k}}(\mu,r)=\delta_{\emph{off,ach}} (\mu,r) + \frac{\mu}{r} \Big ( \frac{Kp}{K(1-p/N)+p}\Big ), 
\end{aligned}
\end{equation}
where $\delta_{\emph{off,ach}} (\mu,r )$ is the offline achievable NDT. 
\end{prop}
\begin{proof}
The result follows by analyzing the Markov chain that describes the number of popular cached files which in turn contributes to the second term in \eqref{k_offcent_time}. Details can be found in Appendix \ref{known_lem_proof}.
\end{proof}
We now propose an improvement that is based on an \textit{adaptive} choice of the file fraction to be cached, as a function of the probability $p$ of new file, as well as the fractional cache size $\mu$ and the fronthaul rate $r$. The main insight here is that, if the probability $p$ is large, it is preferable to cache a fraction smaller than $\mu$ when the resulting fronthaul overhead offsets the gain accrued by means of caching. It is emphasized that caching a fraction smaller than $\mu$ entails that the ENs' cache capacity is partially unused.        
\begin{prop}\label{known_react}
For an $M \times K$ F-RAN with $N \geq  K $, in the known popular set case, online reactive caching with adaptive fractional caching achieves the following long-term NDT
\begin{align}\label{adNDTser}
\bar{\delta}_{\emph{react,adapt,k}}(\mu,r)  
&=\left\{
\begin{array}{ll}
\bar{\delta}_{\emph{react,k}}(\mu,r) &  \text{if}~ p \leq p_0(\mu,r)   \\
\bar{\delta}_{\emph{react,k}}(1/M,r)  &  \text{if}~ p_0(\mu,r) \leq p \leq p_1(\mu,r)    \\
\bar{\delta}_{\emph{C-RAN}}(r) &  \text{if}~  p_1(\mu,r) \leq p \leq 1  
\end{array} \right. 
\end{align}
where  probabilities $p_0(\mu,r)$ and $p_1(\mu,r)$ satisfy $p_0(\mu,r) \leq p_1(\mu,r)$ and the full expressions are given in Appendix \ref{pknown_react}.  Furthermore, we have the inequalities $\bar{\delta}_{\emph{react,adapt,k}}(\mu,r) \leq \bar{\delta}_{\emph{react,k}}(\mu,r)$ with equality if and only if $ p \leq p_0(\mu,r)$. 
\end{prop}
\begin{proof}
See Appendix \ref{pknown_react}. 
\end{proof}

As anticipated, Proposition \ref{known_react} is based on the observation that, when the probability of a new file is larger than a threshold, identified as $p_1(\mu,r)$ in \eqref{adNDTser}, it is preferable not to update the caches and simply use C-RAN delivery. Instead, with moderate probability $p$, i.e., $ p_0(\mu,r) \leq p \leq p_1(\mu,r)$, the performance of the reactive scheme in Proposition \ref{known_lem} is improved by caching a smaller fraction than $\mu$, namely $1/M \leq \mu$. Finally, when $p \leq p_0(\mu,r)$, no gain can be accrued by caching a fraction smaller than $\mu$. This observation is illustrated in Fig. \ref{EX}, which shows the NDT of reactive caching in the known popular set case, without adaptive caching (Proposition \ref{known_lem}) and with adaptive caching (Proposition \ref{known_react}) in the top figure and the corresponding cached fraction in the bottom figure for $M=10$, $K=N=5$, $r=1.1$ and $\mu=0.5$. 

\begin{figure}\label{EX}
\centering
\subfigure[]{
\includegraphics[width=.7\textwidth]{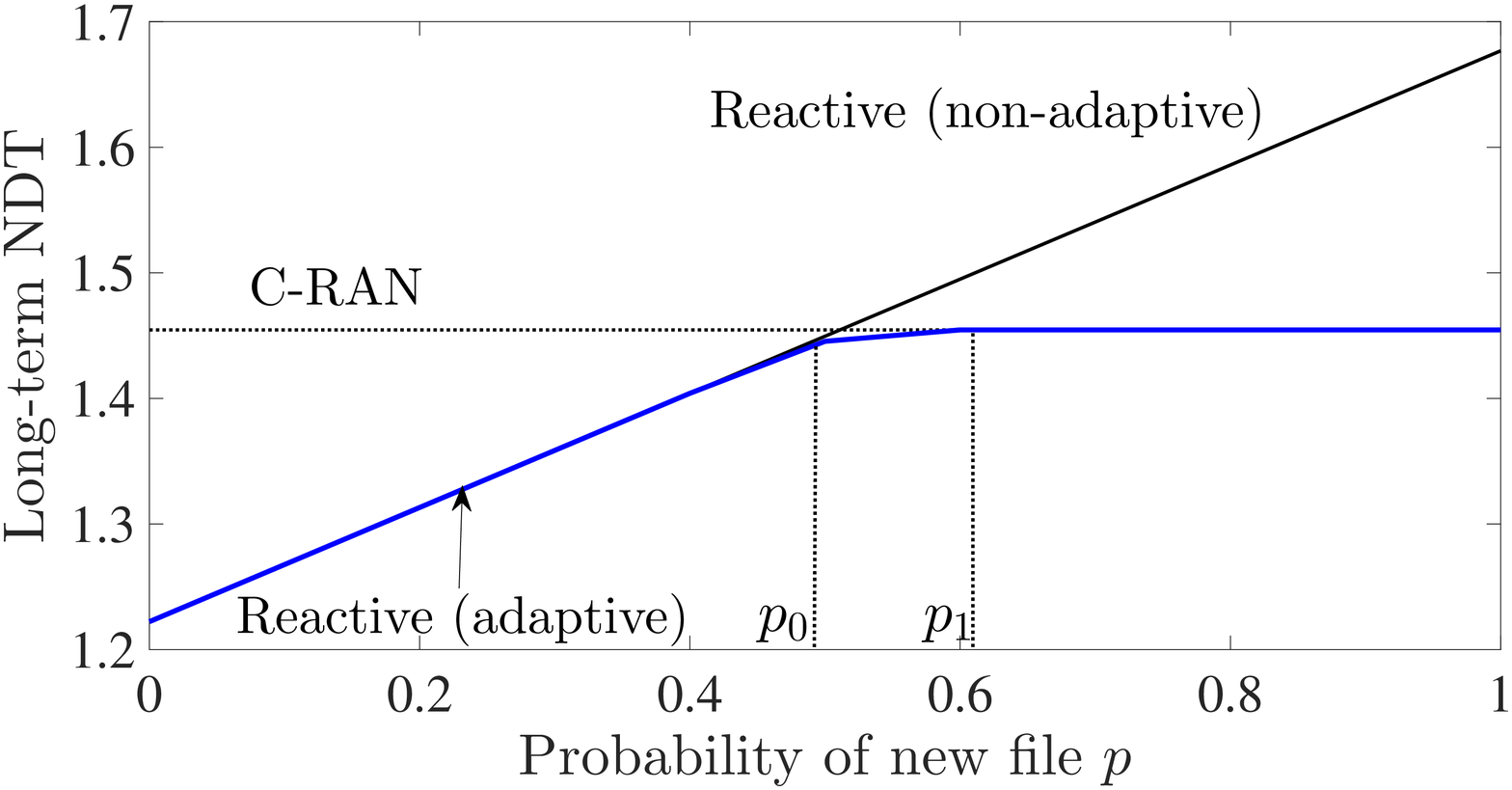}
\vspace{-10mm}} 
\subfigure[]{
\includegraphics[ width=.7\textwidth]{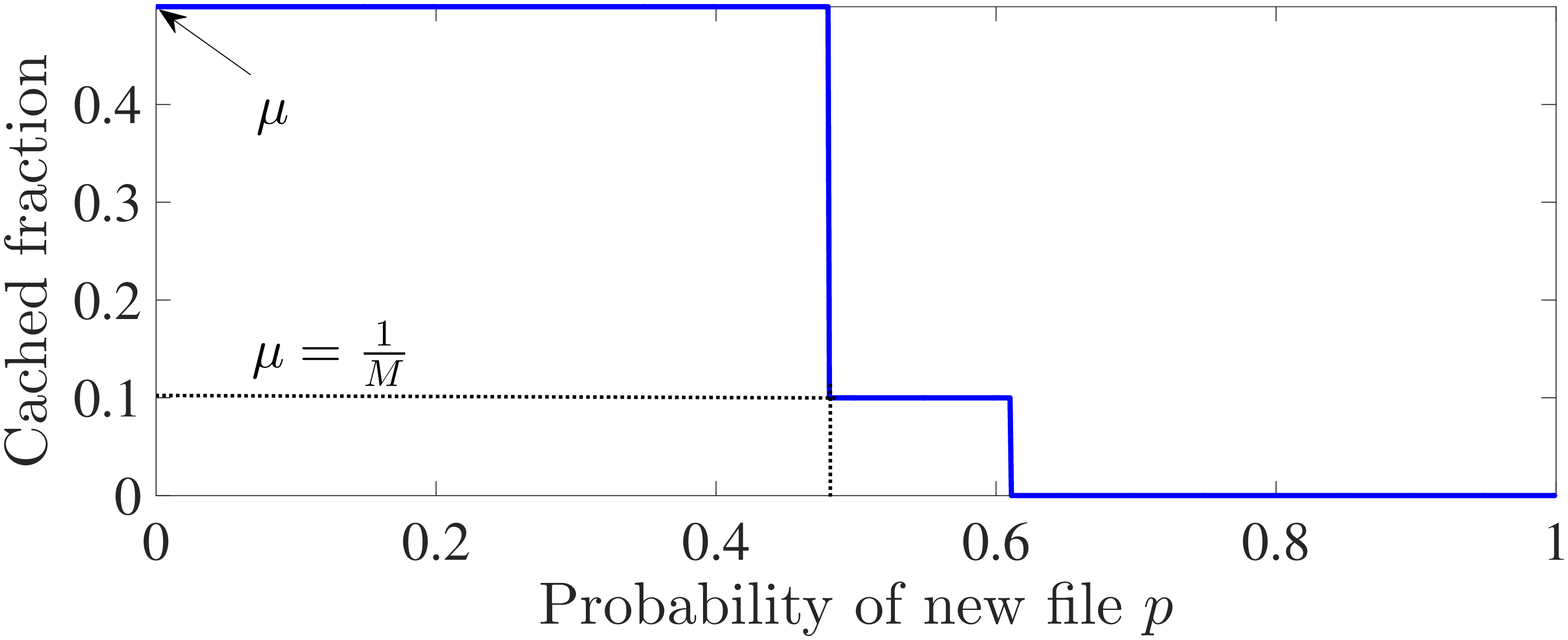}
} 
\caption{Reactive online caching with known popular set under non-adaptive caching (Proposition \ref{known_lem}) and adaptive caching  (Proposition \ref{known_react}) with $M=10$, $K=N=5$, $r=1.1$ and $\mu=0.5$: (a) NDT; (b) fraction cached by the adaptive caching scheme. }
\label{EX}
\end{figure}

\subsection{Reactive Online Caching with Unknown Popular Set}\label{un_re}
In the absence of knowledge about the popular set $\mathcal{F}_t$, the cloud cannot instruct the ENs about which files to evict while guaranteeing that no popular files will be removed from the caches. To account for this constraint, we now consider a reactive caching scheme whereby the ENs evict from the caches a randomly selected file. Random eviction can be improved by other eviction strategies, which are more difficult to analyze, as discussed in Sec. \ref{Numer}.

To elaborate, as pointed out in \cite{Ramtin}, in order to control the probability of evicting a popular file, it is useful for the ENs to cache a number of $N^{'}=\alpha N$ files for some $\alpha >1$. Note that in general the set of $N^{'} > N$ cached files in the cached contents $S_{m,t}$ of all ENs $m$ generally contains files that are no longer in the set $\mathcal{F}_t$ of $N$ popular files. 

If $R_t$ requested files, with $0 \leq  R_t \leq K$, are not cached at the ENs, we propose to transfer a $\mu / \alpha$-fraction of each requested and uncached file on the fronthaul link to each EN by following the offline caching policy in Fig. \ref{seroff} with $\mu / \alpha$ in lieu of $\mu$. Caching a fraction $\mu/\alpha$ is necessary in order to satisfy the cache constraint given the larger number $N'$ of cached files. Adaptive caching is also possible but not covered here due to space limitation. Delivery then takes place via the achievable offline delivery strategy reviewed in Fig. \ref{seroff}, with the only caveat that $\mu / \alpha$ should replace $\mu$. 

The overall NDT is hence the sum of the NDT $\delta_{\text{off,ach}}(\mu / \alpha,r)$ achievable by the offline delivery policy when the fractional cache size is $\mu / \alpha$ and of the NDT due to the  fronthaul transfer of the $\mu / \alpha$-fraction of each requested and uncached file on the fronthaul link. By \eqref{FNDT}, the latter equals $((\mu / \alpha)/C_F) \times \log P=\mu / (\alpha r)$, and hence the overall achievable NDT at each time slot $t$ is
\begin{equation}\label{offcent_time}
\delta_{t} (\mu,r)=\delta_{\text{off,ach}} \Big (\frac{\mu}{\alpha},r \Big )+\frac{\mu}{\alpha}\Big (\frac{\textrm{E}[R_t]}{r} \Big).
\end{equation}
The following proposition presents an achievable long-term NDT for the proposed reactive online caching policy.
\begin{prop}\label{fup1}
\noindent For an $M \times K$ F-RAN with $N \geq K $, in the unknown popular set case, the online reactive caching scheme achieves the long-term NDT that is upper bounded as
\begin{equation}
\begin{aligned}\label{pre_final}
 \bar{\delta}_{\emph{react,u}}(\mu,r) \leq \delta_{\emph{off,ach}} \Big(\frac{\mu}{\alpha},r\Big) + \frac{p\mu}{r(1-p/N)(\alpha-1)},
\end{aligned}
\end{equation}
where $\delta_{\emph{off,ach}} (\mu,r)$ is offline achievable NDT and $\alpha > 1$ is an arbitrary parameter.
\end{prop}
\begin{proof}
Plugging the achievable NDT \eqref{offcent_time} into the definition of long-term NDT in \eqref{avNDT}, we have
\begin{equation}
\begin{aligned}\label{bound_offcent_avge}
 \bar{\delta}_{\text{react,u}}(\mu,r)= \delta_{\text{off,ach}} \Big (\frac{\mu}{\alpha},r \Big )+ \Big ( \frac{\mu}{\alpha r} \Big )\underset{T \rightarrow \infty}{\lim \sup} \frac{1}{T}\sum_{t=1}^{T} \textrm{E}[R_t]. \\
\end{aligned}
\end{equation}
Furthermore, because of the users' demand distribution, caching and random eviction policies are the same as in \cite{Ramtin}, we can leverage \cite[Lemma 3]{Ramtin} to obtain the following upper bound on the long-term average number of requested but not cached files as
\begin{equation}
\begin{aligned}\label{rami}
\underset{T \rightarrow \infty}{\lim \sup} \frac{1}{T}\sum_{t=1}^{T} \textrm{E}[R_t] \leq \frac{p}{(1-p/N)(1-1/\alpha)}.
\end{aligned}
\end{equation}
Plugging \eqref{rami} into \eqref{bound_offcent_avge} completes the proof.
\end{proof}
We emphasize that the right-hand side of \eqref{bound_offcent_avge} is an upper bound on an achievable NDT and hence it is also achievable. 

\section{Pipelined fronthaul-edge transmission}\label{sec_para}
As an alternative to the serial delivery model discussed in Sec. \ref{proform}, in the pipelined fronthaul-edge transmission model, the ENs can transmit on the wireless channel while receiving messages on the fronthaul link. Simultaneous transmission is enabled by the orthogonality of fronthaul and edge channels. Intuitively,  pipelining can make caching  more useful in reducing the transmission latency with respect to serial delivery. In fact, with pipelining, while the ENs transmit the cached files, they can receive the uncached information on the fronthaul links at no additional cost in terms of latency. Pipelined delivery was studied in \cite{Avik2} under offline caching.  With online caching, as we will discuss here, pipelined fronthaul-edge transmission creates new opportunities that can be leveraged by means of proactive, rather than reactive, caching. We recall that proactive caching entails the storage of as-of-yet unrequested files at the ENs. 

In the following, we first discuss the system model for pipelined transmission, then we review results for offline caching from \cite{Avik2}, and finally propose both reactive and proactive online caching policies. The analysis leverages the serial delivery techniques proposed in Sec. \ref{proform} as building blocks.
\subsection{System Model}\label{sysmo2}
The system model for pipelined fronthaul-edge transmission follows in Sec. \ref{sysmo1} with the following differences. First, as discussed, each EN can transmit on the edge channel and receive on the fronthaul link at the same time. Transmission on the wireless channel can hence start at the beginning of the transmission interval, and the ENs use the information received on the fronthaul links in a causal manner. Accordingly, at any time instant $l$ within a time slot $t$, the edge transmission policy of EN $m$ 
maps the demand vector $d_t$, the global CSI $H_{t}$, the local cache content $S_{m,t}$ and the fronthaul messages $U_{m,t,l'}$ received at previous instants $l' \leq l-1$, to the transmitted signal $X_{m,t,l}$ at time $l$. 

Second, the NDT performance metric needs to be adapted. To this end, we denote the overall transmission time in symbols within slot $t$ as $T^{pl}_{t}$, where the superscript  \enquote{pl}  indicates pipelined transmission. For a given sequence of feasible policies, the average achievable delivery time per bit in slot $t$ is defined as
\begin{equation}\label{DTB_P}
\Delta^{pl}_t(\mu,C_F,P)=\underset{L \rightarrow \infty}{\lim}\frac{1}{L}\textrm{E}  [T^{pl}_{t}], \\
\end{equation}
where the average is taken with respect to the distributions of the random variables $\mathcal{F}_t$, $d_{t}$ and $H_{t}$.
The corresponding NDT achieved at time slot $t$ is 
\begin{equation}\label{NDT_P}
\delta^{pl}_t(\mu,r)=\underset{P \rightarrow \infty}{\lim}\frac{\Delta^{pl}_t(\mu,r \log P,P)}{1/\log(P)}
\end{equation}
and, the \textit{long-term NDT} is defined as
\begin{equation}\label{avNDT_P}
\bar{\delta}^{pl}_{\text{on}}(\mu,r)=\underset{T \rightarrow \infty}{\lim \sup} \frac{1}{T}\underset{t=1}{\overset{T}{\sum}} \delta^{pl}_t(\mu,r).
\end{equation}
We denote the minimum long-term NDT over all feasible policies under the known popular set assumption as $\bar{\delta}^{pl^*}_{\text{on,k}}(\mu,r)$, while $\bar{\delta}^{pl^*}_{\text{on,u}}(\mu,r)$ indicates the minimum long-term NDT under the unknown popular set assumption.
As a benchmark, we also consider the minimum NDT
for offline edge caching $\delta_{\text{off}}^{pl^*}(\mu,r)$ as studied in \cite{Avik2}. By construction, we have the inequalities $\delta_{\text{off}}^{pl^*}(\mu,r) \leq \bar{\delta}_{\text{on,k}}^{pl^*}(\mu,r) \leq \bar{\delta}_{\text{on,u}}^{pl^*}(\mu,r)$. Furthermore, while pipelining generally improves the NDT performance as compared to serial delivery, the following result demonstrates that the NDT can be reduced by at most a factor of $2$.   
\begin{lemma}\label{POrderOptimalityOFFMIN}
\noindent  For an $M \times K$ F-RAN with  $N \geq K $, the minimum long-term NDT under online caching with pipelined fronthaul-edge transmission satisfies 
\begin{equation}
\begin{aligned}\label{achtomin2}
\bar{\delta}^{pl^*}_{\emph{on}}(\mu,r) \geq \frac{1}{2}\bar{\delta}^{*}_{\emph{on}}(\mu,r),
\end{aligned}
\end{equation}
where $\bar{\delta}^{*}_{\emph{on}}(\mu,r)$ is the minimum long-term NDT of online caching under serial fronthaul-edge transmission. 
\end{lemma}
\begin{proof}
Consider an optimal policy for pipelined transmission. We show that it can be turned into a serial policy with a long-term NDT which is at most double the long-term NDT for the pipelined scheme. To this end, it is sufficient for the ENs to start transmission on the edge after they receive messages on the fronthaul links. The resulting NDT in each time slot is hence at most twice the optimal long-term NDT of pipelined transmission, since in the pipelined scheme both  fronthaul and edge transmission times are bounded by the overall latency.
\end{proof}

\begin{figure}[t]\label{ref_off}
\vspace{-5mm}
\centering
\includegraphics[width=.5\textwidth]{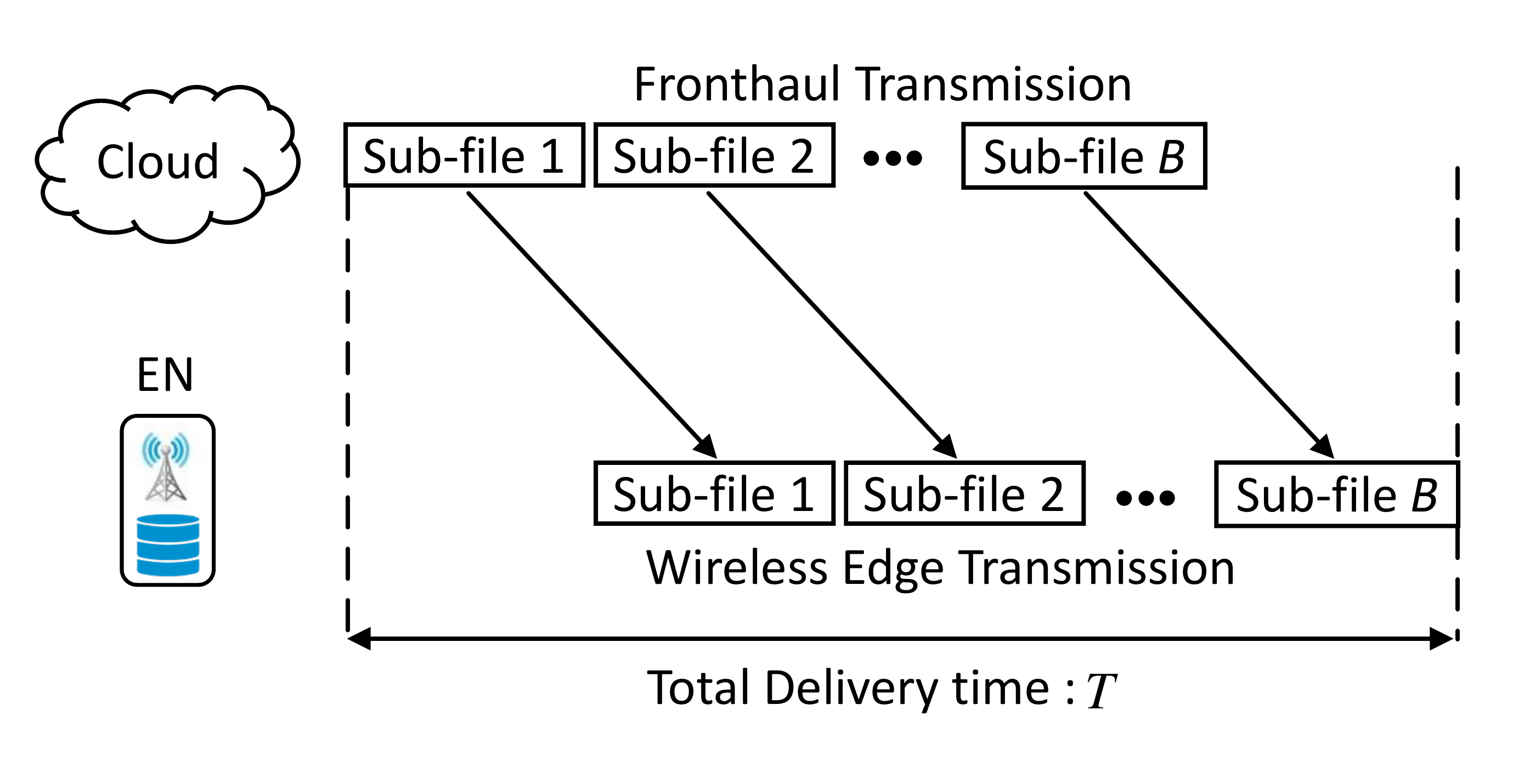}
\caption{Block-Markov encoding converts a serial fronthaul-edge transmission policy into a pipelined transmission policy.}
\label{ONex}
\vspace{-6mm}
\end{figure}
\subsection{Preliminaries}\label{pre_para}
As done for serial transmission, we first review the existing results for offline caching with pipelined transmission. 
The achievable scheme for offline caching proposed in \cite{Avik2} utilizes \textit{block-Markov encoding}, which is illustrated in Fig. \ref{ONex}. 
Block-Markov encoding converts a serial fronthaul-edge strategy into a pipelined scheme. To this end, each file of $L$ bits is divided into $B$ sub-files, each of $L/B$ bits. The transmission interval is accordingly divided into $B+1$ sub-frames. The first sub-file is sent in the first sub-frame on the fronthaul, and then on the wireless link in the second sub-frame. During the transmission of first sub-file on the wireless link, the second sub-file is sent on the fronthaul link in the second sub-frame. As illustrated in Fig. \ref{ONex}, the same transmission scheme is used for the remaining sub-blocks. 

If the original serial transmission scheme has edge-NDT $\delta_{E}$ and fronthaul-NDT $\delta_{F}$, then the NDT of pipelined scheme when $B \rightarrow \infty$ can be proved to be \cite{Avik2}
\begin{equation}\label{def_pipe_off}
\delta^{pl}_{\text{off}}=\max \big ( \delta_{E},\delta_{F} \big ).
\end{equation}
This is because the maximum of the latencies of the simulatneously occuring fronthaul and edge transmissions determines the transmisson time in each time slot. In \cite{Avik2}, an achievable NDT $\delta^{pl}_{\text{off,ach}}(\mu,r)$ is obtained by using block Markov encoding along with file splitting and time sharing using the same constituent schemes considered in Sec. \ref{OFF} for serial transmission. We refer for details to \cite[Sec. VI-B]{Avik2}. We now discuss online caching strategies.

\subsection{C-RAN Delivery}\label{PGriddy}
Similar to Sec. \ref{Griddy}, with C-RAN delivery, the long-term NDT coincides with the NDT obtained in each time slot using in \eqref{def_pipe_off} the edge and fronthaul-NDTs in \eqref{cran_e}, yielding
\begin{equation}
\begin{aligned}\label{pgreedy}
\bar{\delta}^{pl}_{\text{C-RAN}}(\mu,r)= \max \Bigg ( \frac{K}{\min(M,K)}, \frac{K}{Mr} \Bigg ). 
\end{aligned}
\end{equation}

\subsection{Reactive Online Caching} \label{Rea_P}
We now briefly discuss reactive caching policies that extend the approach in Sec. \ref{ReactK_S} and Sec. \ref{un_re} to pipelined transmission. We recall that, with reactive caching, the requested files that are not partially cached at ENs are reactively sent on the fronthaul. Furthermore, if the set of popular files is known, the ENs can evict the unpopular cached files to make space for newly received files. Otherwise, in the unknown popular set case, a randomly selected file can be evicted.  As for serial transmission, we propose to choose the fraction to be cached and to perform delivery by following the offline caching strategy of \cite{Avik2}. Obtaining  closed forms for the achievable long-term NDTs appears to be prohibitive.  
In Sec. \ref{Numer}, we provide Monte Carlo simulation to numerically illustrate the performance of the proposed scheme. 
 
\subsection{Proactive Online Caching} \label{Pro_P}
With pipelined transmission, by \eqref{def_pipe_off}, when the rate of the fronthaul link is sufficiently large, the NDT in a time slot may be limited by the edge transmission latency. In this case, the fronthaul can be utilized to send information even if no uncached file has been requested without affecting the NDT. This suggests that proactive caching, whereby uncached and unrequested files are pushed to the ENs, potentially advantageous over reactive caching. Note that this is not the case for serial transmission, whereby any information sent on the fronthaul contributes equally to the overall NDT, irrespective of the edge latency, making it only useful to transmit requested files on the fronthaul links\footnote{This corrects an erroneous statement in \cite{Azimi}.}. 

To investigate this point, here we study a simple  proactive online caching scheme. Accordingly, every time there is a new file in the popular set,  a $\mu$-fraction of the file is  proactively sent on the fronthaul links in order to update the ENs' cache content. This fraction is selected to be distinct across the ENs if $\mu \leq 1/M$, hence enabling EN coordination (recall the top-left part of Fig. \ref{seroff}); while the same fraction $\mu$ is cached at all ENs otherwise, enabling delivery via EN cooperation (see bottom of Fig. \ref{seroff}). 
\begin{prop}\label{p_proact}
For an $M \times K$ F-RAN with $N \geq K $ and pipelined transmission, proactive online caching achieves the long-term NDT
\begin{align}
\bar{\delta}^{pl}_{\emph{proact}}(\mu,r)
& =  p \max \Big \{(\mu M)\delta_{\emph{F,Coor}}+(1-\mu M)\delta_{\emph{F,C-RAN}}(r)+ \frac{\mu }{r}, (\mu M) \delta_{\emph{E,Coor}}+(1-\mu M)\delta_{\emph{E,C-RAN}}(r) \Big \} \nonumber \\  
&+ (1-p) \max \Big \{(\mu M)\delta_{\emph{F,Coor}}+(1-\mu M)\delta_{\emph{F,C-RAN}}(r), (\mu M) \delta_{\emph{E,Coor}}+(1-\mu M)\delta_{\emph{E,C-RAN}}(r) \Big \},  \label{pproach1}
\end{align}
for $\mu \in [0,1/M]$, and
\begin{align}
\bar{\delta}^{pl}_{\emph{proact}}(\mu,r)
& =  p \max \Big \{\mu \delta_{\emph{F,Coop}}+(1-\mu)\delta_{\emph{F,C-RAN}}(r)+ \frac{\mu }{r}, \mu \delta_{\emph{E,Coop}}+(1-\mu)\delta_{\emph{E,C-RAN}}(r) \Big \} \nonumber \\  
  &+ (1-p) \max \Big \{\mu \delta_{\emph{F,Coop}}+(1-\mu)\delta_{\emph{F,C-RAN}}(r), \mu \delta_{\emph{E,Coop}}+(1-\mu)\delta_{\emph{E,C-RAN}}(r) \Big \},  \label{pproach2}
\end{align} 
for $\mu \in [1/M,1]$,
with the definitions given in  \eqref{coop_e}-\eqref{cran_e}.
\end{prop}
\begin{proof}
With probability of $(1-p)$, the popular set remains unchanged and the NDT in the given slot is obtained by time sharing between C-RAN delivery, for the uncached $(1-\mu)$-fraction of the requested files, and either EN coordination and EN cooperation, depending on the value of $\mu$ as discussed above, for the cached $\mu$-fraction. Instead, with probability $p$, there is a new file in the popular set, and a $\mu$ fraction of file is proactively sent on the fronthaul links, resulting in an additional term $\mu/r$ to the fronthaul-NDT of offline schemes. 
\end{proof}

\section{Impact of time-varying popularity}\label{Comp}
In this section, we compare the performance of offline caching in the presence of a static set of popular files with the performance of online caching under the considered dynamic popularity model. The analysis is intended to bring insight into  the impact of a time-varying popular set on the achievable delivery latency.  We focus here on the case in which the number of ENs is larger than the number of users, namely, $M \geq K$. 
\begin{prop}\label{final_uppro1}
\noindent For an $M \times K$ F-RAN and $N > M \geq K \geq 2 $ and $r >0$, under both serial and pipelined delivery modes with known and unknown popular set, the minimum long-term NDT $\bar{\delta}_{\emph{on}}^*(\mu,r)$ satisfies the condition
\begin{equation}\label{longpro1}
\bar{\delta}_{\emph{on}}^*(\mu,r) = c\delta^*_{\emph{off}} (\mu,r) +  O \Big (\frac{1}{r} \Big ),
\end{equation}
where $ c \leq 4$ is a constant and $\delta^*_{\emph{off}} (\mu,r)$ is the minimum NDT under offline caching.  
\end{prop}
\begin{proof}
See Appendix \ref{fund_r_pro}.
\end{proof}

Proposition \ref{final_uppro1} shows that the long-term NDT with online caching is proportional to the minimum NDT for offline caching and static popular set, with an additive gap that is inversely proportional to the fronthaul rate $r$. 
To see intuitively why this result holds, note that, when $\mu \geq 1/M$ and hence the set of popular files can be fully stored across all the $M$ EN's caches, offline caching enables the delivery of all possible users' requests with a finite delay even when $r=0$. In contrast, with online caching, the time variability of the set $\mathcal{F}_t$ of popular files implies that, with  non-zero probability, some of the requested files cannot be cached at ENs and hence should be delivered by leveraging fronthaul transmission. Therefore, the additive latency gap as a function of $r$ is a fundamental consequence of the time-variability of the content set.

\section{Numerical Results}\label{Numer}
In this section, we complement the analysis of the previous sections with numerical experiments.  We consider in turn serial transmission, as studied in  Sec. \ref{Fundamental limits}, and  pipelined transmission covered in Sec. \ref{sec_para}. 

\begin{figure}
\vspace{2mm}
\centering
\includegraphics[width=.7\textwidth]{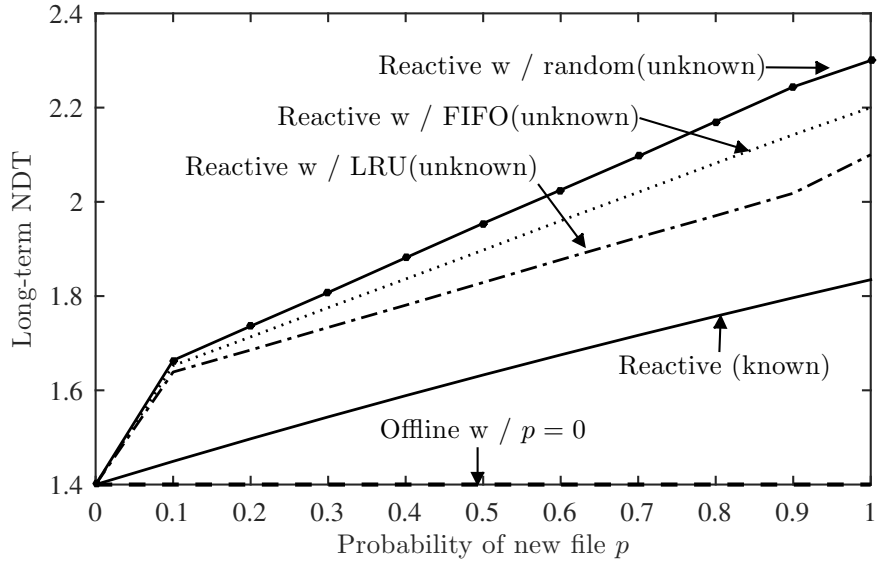}
\caption{Long-term NDT of reactive online caching with known popular set, as well as reactive online caching with unknown popular set using different eviction policies ($M=10$, $K=5$, $N=20$, $\mu=0.1$ and $r=0.2$). }
\vspace{-5.5 mm}
\label{F1}
\end{figure}
\textit{Serial delivery}: For serial transmission, we consider the performance of reactive online caching with known popular set (eq. \eqref{react_kkk}) and unknown popular set (bound in  \eqref{pre_final}).  For the latter, we evaluate the NDT via Monte Carlo simulations by averaging over a large number of realizations of the random process $R_t$ of requested but uncached files (see \eqref{bound_offcent_avge}), which is simulated starting from empty caches at time $t=1$. We also plot the NDT of offline scheme of \cite{Avik2} described in Sec. \ref{OFF} in the presence of a time-invariant popular set. 
 
We first consider the impact of the rate of change of the popular content set for $\mu=0.1$ and $r=0.2$. To this end, the long-term NDT is plotted as a function of the probability $p$ in Fig. \ref{F1}. We observe that variations in the set of popular files entail a performance loss of online caching with respect to offline caching with static popular set that increases with $p$. Furthermore, under random eviction, the lack of knowledge of the popular set is seen to cause a significantly larger NDT than the scheme that can leverage knowledge of the popular set. This performance gap can be reduced by using better eviction strategies.

To investigate this point, we evaluate also the Monte Carlo performance of reactive online caching with unknown popular set under the following standard eviction strategies: Least Recently Used (LRU), whereby the replaced file is the one that has been least recently requested by any user; and First In First Out (FIFO), whereby the file that has been in the caches for the longest time is replaced. 
From Fig. \ref{F1}, LRU and FIFO are seen to be both able to improve over randomized eviction, with the former generally outperforming the latter, especially for large values of $p$. Finally, we note that for  the long-term NDT of C-RAN (eq. \eqref{greedy}) is constant for all values of $p$ and equal to $7.5$  (not shown). 

The impact of the fronthaul rate is studied next by means of Fig. \ref{F2}, in which the long-term NDT is plotted as a function of $r$ for $\mu=0.5$ and  $p=0.8$.  The main observation is that, as $r$ increases, delivering files from the cloud via fronthaul resources as in C-RAN yields decreasing latency losses, making edge caching less useful. In contrast, when $r$ is small, an efficient use of edge resources via edge caching becomes critical, and the achieved NDT depends strongly on the online cache update strategy. 

\begin{figure}
\centering
\includegraphics[width=.7\textwidth]{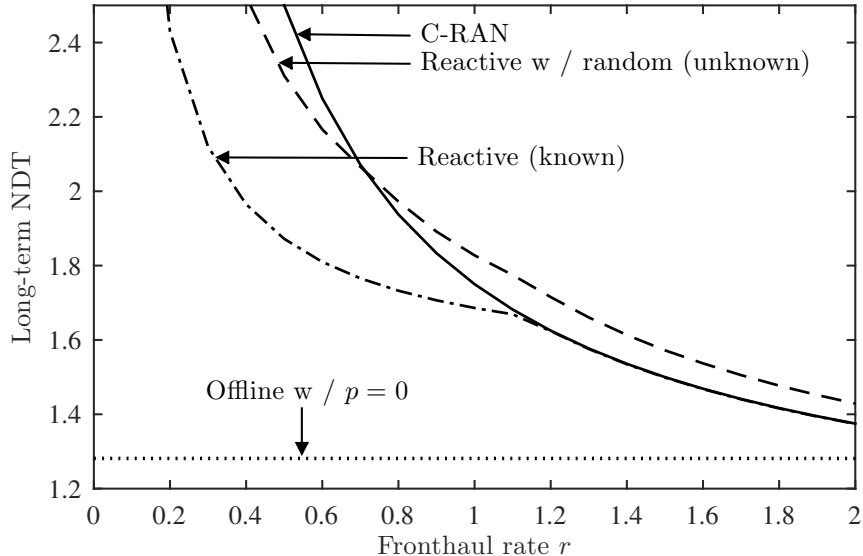}
\caption{Long-term NDT of reactive online caching with known and unknown popular set, as well as C-RAN transmission and offline caching, under serial fronthaul-edge transmission ($M=4$, $K=3$, $N=10$, $\mu=0.25$ and $p=0.9$). }
\vspace{-5.5 mm}
\label{F2}
\end{figure}  

\textit{Pipelined delivery}: We now evaluate the long-term NDT performance  of pipelined fronthaul-edge transmissions in Fig. \ref{F3}. The NDTs are computed using Monte Carlo simulations as explained above with reactive and proactive caching under known popular set. The plot shows the long-term NDT as a function of the fronthaul rate $r$. It is observed that, for all schemes, when the fronthaul rate is sufficiently large, the long-term NDT is limited by the edge-NDT (recall \eqref{def_pipe_off}) and hence increasing $\mu$ cannot reduce the NDT. In contrast, for smaller values of $r$, caching can decrease the long-term NDT. In this regime, proactively updating the ENs' cache content can yield a lower long-term NDT than reactive schemes, whereby the fronthaul links are underutilized, as discussed in Sec \ref{Rea_P}.  
\begin{figure}[t]
\centering
\includegraphics[width=.7\textwidth]{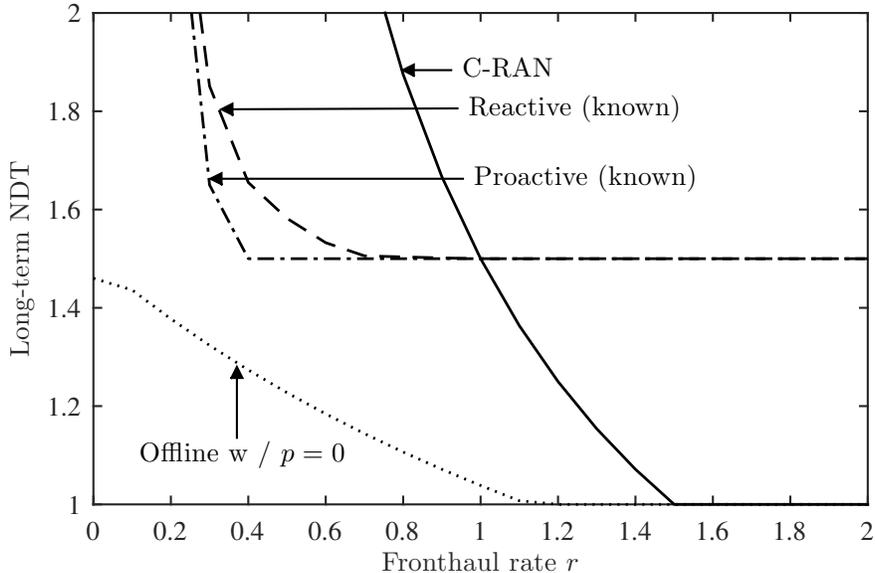}
\caption{Long-term NDT of reactive and proactive online caching with known popular set, as well as C-RAN transmission and offline caching, under pipelined fronthaul-edge transmission ($M=4$, $K=3$, $N=10$, $\mu=0.25$ and $p=0.9$). }
\vspace{-5.5 mm}
\label{F3}
\end{figure}

\textit{Serial vs. pipelined delivery}: We now compare the long-term NDT performance  of pipelined and serial fronthaul-edge transmission by means of Fig. \ref{F4}. The plot shows the long-term NDT of reactive online schemes under unknown popular set for serial transmission (eq. \eqref{bound_offcent_avge}) as well as pipelined transmission (see Sec. \ref{Rea_P}) as a function of the fractional cache size $\mu$ for $M=2$, $N=10$, $r=0.2$ and $p=0.8$.  For both cases, we evaluate the NDT via Monte Carlo simulations by averaging over a large number of realizations of the random process $R_t$ of requested but uncached files.  As discussed in Sec. \ref{Pro_P}, pipelined transmission generally provides a smaller long-term NDT, and the gain in terms of NDT is limited to a factor of at most two. It is observed that, following the discussion in Sec. \ref{Pro_P}, edge caching enables larger gains in the presence of pipelined transmission owing to the capability of the ENs to transmit and receive at the same time. In particular, when $\mu$ is large enough, the achieved long-term NDT coincides with that of offline caching, indicating that, in this regime, the performance is dominated by the bottleneck set by wireless edge transmission.

\begin{figure}[t]
\vspace{-1mm}
\centering
\includegraphics[width=.7\textwidth]{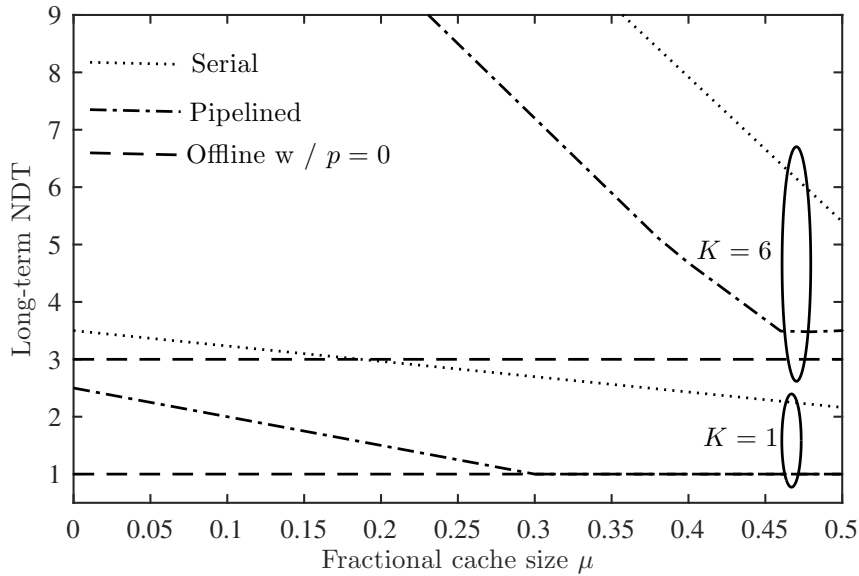}
\caption{Long-term NDT of reactive online caching with known popular set for serial and pipelined transmission, as well as of offline caching ($M=2$, $N=10$, $r=0.2$ and $p=0.8$). }
\vspace{-5.5 mm}
\label{F4}
\end{figure}

\section{Conclusions} \label{Conc}
In this work, we considered the problem of content delivery in a fog architecture with time-varying content popularity. In this setting, online edge caching is instrumental in reducing content delivery latency. For the first time, an information-theoretic analysis is provided to obtain insights into the optimal use of fog resources, namely fronthaul link capacity, cache storage at edge, and wireless bandwidth. The analysis adopts a high-SNR latency metric that captures the performance of online caching design. 
We first studied a serial transmission mode whereby the ENs start transmission on the wireless channel after completion of cloud-to-EN transmission. Then, we investigated a pipelined delivery mode that allows for simultaneous transmission and reception of information by the ENs. In both cases, the impact of knowledge of the set of popular files was considered. 

One of the main insights of the analysis is that, regardless of the cache capacity at the edge, of prior knowledge at the cloud about the time-varying set of popular contents, and of the implementation of serial or pipelined delivery mode, the rate of the fronthaul links sets a fundamental limit on the achievable latency performance. We have specifically shown that the additive performance gap caused by changes in popularity is inversely proportional to the fronthaul rate, given that new content needs to be delivered through the fronthaul links (Proposition \ref{final_uppro1}). This conclusion highlights the increasing relative importance of provisioning sufficient fronthaul capacity as opposed to additional cache storage when the rate of change of the popular files increases. The fronthaul capacity requirements can be reduced if the set of popular files can be estimated (Fig. \ref{F2}). Furthermore, in order to offset the fronthaul overhead when the rate of change of popularity is sufficiently large, it is preferable to cache a smaller fraction of the new files than allowed by the cache capacity (Proposition \ref{known_react} and Fig. \ref{EX}).

Another important design insight is that, by endowing the ENs with the capability to transmit simultaneously on fronthaul and edge channels, one gains the capability to better leverage the available fronthaul resources by means of proactive online caching. In fact, proactive caching can opportunistically exploit unused fronthaul transmission capacity, particularly when the cache capacity is large enough (Fig. \ref{F3} and Fig. \ref{F4}). Moreover, with pipelined transmission, the latency is limited by either edge or fronthaul latencies. As a result, for a system with sufficient fronthaul and cache resources, the performance can match that with a static set of popular files and full caching (Proposition \ref{p_proact} and Fig. \ref{F4}).

Finally, our results emphasize the practical importance of deploying content eviction mechanisms such as LRU to bridge to some extent, the gap between the performance under known and unknown popular set (Fig. \ref{F1}).    

Among directions for future work, we mention here the analysis of more general caching strategies that allow for coding and arbitrary caching patterns across files and ENs; of the impact of imperfect CSI; of partial connectivity (see \cite{chang} and references therein); and of more realistic request models \cite{Leo1}; as well as the investigation of online caching within large fog architectures including metro and core segments \cite{chiang}.

\appendices

\section{Proof of Proposition \ref{known_lem}}\label{known_lem_proof}
The proof follows closely the approach used in \cite[Sec. V-A]{Ramtin}. The goal is to compute the long-term NDT  \eqref{avNDT} using \eqref{k_offcent_time}. We recall that  $R_t \in \{ 0, 1, ...,K \}$ is the number of new files requested by users at time slot $t$ which are not available at ENs' caches. We further define as $X_t \in \{ 0, 1, ...,N \}$ the number of files in $\mathcal{F}_t$ which are available at ENs' caches  at the beginning of time slot $t$, and as $V_t \in \{ 0, 1 \}$ the number of files that were correctly cached at the end of time slot $t$ but are no longer popular at time slot $t+1$. 

Using the above-mentioned random processes, the following update equation for the process $X_t$ holds    
\vspace{-5mm}
\begin{equation}\label{drift2}
X_{t+1}=X_t + R_t-V_t. 
\end{equation}
As in \cite{Ramtin}, it can be shown that  $\{ X_t \}$ is a Markov process with a single ergodic recurrent class consisting the states $\{ K-1, K, ...,N \}$ and transient states  $\{ 0,1,...,K-2 \}$. 
Hence, in the steady state where $\textrm{E}[X_{t+1}]=\textrm{E}[X_t]$, we have the equality 
\vspace{-5mm}
\begin{equation}\label{ffff}
\textrm{E}[R_t]=\textrm{E}[V_t]. 
\end{equation}
Furthermore, since each user requests a file in $\mathcal{F}_t$ according to uniform distribution without replacement, conditioned on $X_t$, the random variable $R_t$ has the expected value
\begin{equation}\label{R_f}
\textrm{E}[R_t|X_t]=K \Big (1-\frac{X_t}{N} \Big).
\end{equation}  
We will use \eqref{ffff} and \eqref{R_f} to compute the expectation $\textrm{E}[R_t]$ in steady state. Given the asymptotic stationarity of $X_t$, and hence of $R_t$, by the standard Cesaro mean argument, the long-term NDT \eqref{k_offcent_time} is finally obtained by substituting in \eqref{avNDT} the steady state mean $\textrm{E}[R_t]$.  

To this end, denoting the number of cached popular files at the end of time slot $t$ as $X'_t$, we have  
\vspace{-5mm}
\begin{equation}\label{R_xx}
X'_t=X_{t+1}+V_t, 
\end{equation}
since the number of cached popular files at the start of time slot $t+1$ is either the same as at the end of time slot $t$ or to has one less file due to arrival of a new file in the popular set. Conditioning on $X'_t$, we have
\vspace{-3mm}
\begin{equation}\label{R_x1}
\text{Pr}(V_t=1|X'_t)=p \frac{X'_t}{N}, 
\end{equation}
since with probability of $p$ there is a new file in the popular set which replaces one of the cached popular files at ENs selected with probability of $X'_t/N$ and these two events are independent of each other.
Taking expectation with respect to  $X'_t$ in \eqref{R_x1} and using the fact that $V_t \in \{0,1\}$, we have 
\vspace{-5mm}
\begin{align}\label{R_x2}
\textrm{E}[V_t]=\textrm{E}[\textrm{Pr}(V_t=1|X'_t)] =p \frac{\textrm{E}[X'_t]}{N} \overset{(a)}{=}p \frac{\textrm{E}[X_{t+1}]+\textrm{E}[V_t]}{N}\overset{(b)}{=}p \frac{\textrm{E}[X_{t}]+\textrm{E}[V_t]}{N}, 
\end{align}
where $(a)$ is obtained using \eqref{R_xx} and $(b)$ is obtained for steady state where $\textrm{E}[X_{t+1}]=\textrm{E}[X_{t}]$. Solving for $E[V_t]$ yields 
\vspace{-3mm}
\begin{equation}\label{RR_x}
\textrm{E}[V_t]=\frac{p}{1-p/N}\frac{\textrm{E}[X_{t}]}{N}.  
\end{equation}
Taking expectation respect to $X_t$ from \eqref{R_f} and then using \eqref{RR_x} and \eqref{ffff},  we have
\begin{equation}\label{R_b}
\textrm{E}[R_t]= K \Big (1-\frac{\textrm{E}[X_{t}]}{N} \Big )= \frac{Kp}{K(1-p/N)+p}.
\end{equation}
Plugging  \eqref{R_b} into \eqref{k_offcent_time} completes the proof. 

\section{proof of proposition \ref{known_react}}\label{pknown_react}
The proposed adaptive caching scheme optimizes the choice of the cached fraction within the reactive scheme achieving $\bar{\delta}_{\text{react,k}}(\mu,r)$ in \eqref{react_kkk}. To this end, it solves the problem  $\bar{\delta}_{\text{react,adapt,k}}(\mu,r)=\underset{ \mu' \leq \mu }{\min} ~ \bar{\delta}_{\text{react,k}}(\mu',r)  $ over the fraction $\mu'$. With some algebra, this optimization yields 
\vspace{-5mm}
\begin{align}
p_0(\mu,r)&= \min \Bigg (\frac{Kr \big ( \min(M,K)-1 \big ) }{K(M-1)+r \big (\min(M,K)-1 \big ) \big (\frac{K}{N}-1 \big )},1 \Bigg ) \nonumber \\
p_1(\mu,r)&=\min \Bigg ( \frac{ K - r \big ( \min(M,K)-1 \big ) }{1+ \Big ( K - r \big ( \min(M,K)-1 \big ) \Big ) \big ( 1/N-1/K\big ) },1 \Bigg)
\end{align}
for $r \leq r_{th}$; and
\vspace{-5mm} 
\begin{align}\label{cons_r_2}
p_0(\mu,r)= p_1(\mu,r)= \min \Bigg (\frac{K}{M+K/N-1},1 \Bigg ),   
\end{align}
for  $r \geq r_{th}$.


\section{proof of proposition \ref{final_uppro1}}\label{fund_r_pro}
To prove Proposition \ref{final_uppro1}, we will show that for serial transmission the following inequalities 
\begin{align}\label{longpro2}
 \frac{1-\frac{Kp}{N}}{2}\delta^*_{\text{off}} (\mu,r) + \frac{Kp}{N}\Big (1+\frac{\mu}{r} \Big) & \leq  \bar{\delta}_{\text{on,k}}^*(\mu,r) \leq \bar{\delta}_{\text{on,u}}^*(\mu,r) \leq  2 \delta^*_{\text{off}} (\mu,r) +  \frac{4}{r},
\end{align}
hold. Comparing the right-most and the left-most expressions will complete the proof. For the pipelined case, we have the following relationships with the minimum NDT for the serial case: 
\begin{align}\label{longpro3}
\bar{\delta}_{\text{on,u}}^{pl^*}(\mu,r) & \leq \bar{\delta}_{\text{on,u}}^*(\mu,r) \overset{(a)}{\leq}  2 \delta^*_{\text{off}} (\mu,r) +  \frac{4}{r} \overset{(b)}{\leq}  4 \delta^{pl^*}_{\text{off}} (\mu,r) +  \frac{4}{r}, 
\end{align}
where $(a)$ is obtained by using the right-most inequality in \eqref{longpro2}, and $(b)$ is obtained from \cite[Lemma 4]{Avik2}, namely from the inequality $\delta_{\text{off}}^{*}(\mu,r) \leq 2\delta_{\text{off}}^{pl^*}(\mu,r)$; and also
\begin{align}\label{longpro4}
\bar{\delta}_{\text{on,u}}^{pl^*}(\mu,r) & \geq  \bar{\delta}_{\text{on,k}}^{pl^*}(\mu,r) \overset{(a)}{\geq} \frac{1}{2} \bar{\delta}_{\text{on,k}}^*(\mu,r) \overset{(b)}{\geq}  \frac{1-\frac{Kp}{N}}{4}\delta^*_{\text{off}} (\mu,r) + \frac{Kp}{2N} \Big (1+\frac{\mu}{r} \Big) , 
\end{align}
where $(a)$ is obtained using Lemma \ref{POrderOptimalityOFFMIN} and $(b)$ is the first inequality in \eqref{longpro2}.

In what follows, we prove \eqref{longpro2} for the serial transmission. 

\textit{Lower bound}:
To prove the lower bound in \eqref{longpro2}, we first present the following lemma, which presents a slight improvement over the lower bound in \cite[Proposition 1]{Avik2}. 
\noindent
\begin{lemma}\label{lboffan}
\noindent (Lower Bound on Minimum offline NDT). For an $M \times K$ F-RAN with $N  \geq K$ files, the minimum NDT is lower bounded as
\begin{equation}
\begin{aligned}\label{off-lb}
\delta_{\emph{off}}^*(\mu,r) \geq \delta_{\emph{off,lb}}(\mu,r)
\end{aligned}
\end{equation}
where $\delta_{\emph{off,lb}}(\mu,r)$ is the minimum value of the following linear program (LP)
\begin{align}
\hspace{-10mm}  \text{minimize}~~& \delta_E+\delta_F \label{LP1} \\
\hspace{-10mm}\text{subject to :}~&l\delta_E+(M-l)r\delta_F \geq K-\min \big ((K-l),(M-l)(K-l)\mu \big ) \label{LP2} \\
\hspace{-10mm}&  \delta_F\geq 0, \delta_E\geq 1, \label{LP3}
\end{align}
where \eqref{LP2} is a family of constraints with $0 \leq l \leq \min \{ M,K\}$.
\end{lemma} 
\begin{proof}
It follows using the same steps as Proposition \ref{onlb} below.  
\end{proof}
Next, we introduce the following lower bound on the minimum long-term NDT of online caching. 
\noindent
\begin{prop}\label{onlb}
\noindent (Lower bound on the Long-Term NDT of Online Caching for Serial Transmission). For an $M \times K$ F-RAN with a fronthaul rate of $r \geq 0$, the long-term NDT is lower bounded as $\bar{\delta}_{\emph{on,u}}^*(\mu,r) \geq \bar{\delta}_{\emph{on,k}}^*(\mu,r)  \geq (1-Kp/N)\delta_{\emph{off,lb}}(\mu,r) + (Kp/N)\delta_{\emph{on,lb}}(\mu,r)$, where $\delta_{\emph{on,lb}}(\mu,r)$ is the solution of following LP
\begin{align}
\text{minimize}~~& \delta_E+\delta_F \label{LP4} \\
\text{subject to :}~&l\delta_E+(M-l)r\delta_F \geq K-\min \Big ((K-l-1), (M-l)(K-l-1)\mu \Big ) \label{LP5}\\
&  \delta_F\geq 0, \delta_E\geq 1, \label{LP6}
\end{align}
where \eqref{LP5} is a family of constraints with $0 \leq l \leq K-1$ and $\delta_{\emph{off,lb}}(\mu,r)$ is the lower bound on the minimum NDT of offline caching defined in Lemma \ref{lboffan}.
\end{prop}
\begin{proof}
See Appendix \ref{lowbd}.
\end{proof}
Now, using Proposition \ref{onlb}, we have
\begin{equation}\label{lbon1}
\begin{aligned}
\bar{\delta}_{\text{on,u}}^*(\mu,r) & \geq \bar{\delta}_{\text{on,k}}^*(\mu,r)  \geq \Big (1-\frac{Kp}{N} \Big )\delta_{\text{off,lb}} (\mu,r) + \frac{Kp}{N}\delta_{\text{on,lb}} (\mu,r)  \overset{(a)}{\geq} \frac{(1-\frac{Kp}{N})}{2}\delta^*_{\text{off}} (\mu,r) + \frac{Kp}{N}\delta_{\text{on,lb}} (\mu,r)  \nonumber \\  
&\overset{(b)}{\geq} \frac{(1-\frac{Kp}{N})}{2}\delta^*_{\text{off}} (\mu,r) + \frac{Kp}{N} \Big ( 1+\frac{\min (\mu,1/M)}{r} \Big ), 
\end{aligned}
\end{equation}
where $(a)$ is obtained using \eqref{achtomin}, namely $\delta^*_{\text{off}} (\mu,r) / \delta_{\text{off,lb}}(\mu,r) \leq 2$ and $(b)$ follows by deriving the lower bound $\delta_F \geq \min (\mu,1/M)/r $ on the optimal solution of the LP \eqref{LP4} by setting $l=0$ in the constraint \eqref{LP5} and summing the result with constraint \eqref{LP6}. 

\textit{Upper bound}: 
To prove the upper bound in \eqref{longpro2}, we leverage the following lemma.
\begin{lemma}\label{up_alpha1}
\noindent For any $ \alpha > 1$, we have the following inequality
\begin{equation}
\begin{aligned}\label{Ach1}
\delta_{\emph{off,ach}}\Big (\frac{\mu}{\alpha},r \Big ) \leq 2\delta_{\emph{off}}^*(\mu,r) + \frac{1}{r} + \frac{1}{\alpha} \Big (1-\frac{1}{r} \Big).
\end{aligned}
\end{equation}
\end{lemma}
\begin{proof}
See Appendix \ref{prolemupa}.
\end{proof}
Using Proposition \ref{fup1} and Lemma \ref{up_alpha1}, an upper bound on the long-term average NDT of the proposed reactive caching scheme is obtained as
\begin{equation}
\begin{aligned}\label{up_a}
\bar{\delta}_{\text{react}}(\mu,r) \leq 2\delta_{\text{off}}^*(\mu,r) + f(\alpha),
\end{aligned}
\end{equation}
where
\begin{equation}
\begin{aligned}\label{Alp}
f(\alpha)=\frac{1}{r} + \frac{1}{\alpha} \Big (1-\frac{1}{r} \Big)+\frac{Np(\mu/r)}{(N-p)(\alpha-1)}.
\end{aligned}
\end{equation}
Since the additive gap \eqref{Alp} is a decreasing function of $N$ and an increasing function of $p$ and $\mu$,  it can be further upper bounded by setting $N=2$, $p=1$ and $\mu=1$. Finally, by plugging $\alpha=2$, and using the inequality $\bar{\delta}_{\text{on,u}}^*(\mu,r) \leq \bar{\delta}_{\text{react}}(\mu,r) $ the upper bound  is proved.

\section{proof of proposition \ref{onlb}}\label{lowbd}
To obtain a lower bound on the long-term NDT, we consider a genie-aided system in which, at each time slot $t$, the ENs are provided  with the optimal cache contents of an offline scheme tailored to the current popular set $\mathcal{F}_t$ at no cost in terms of fronthaul latency. In this system, as in the system under study, at each time slot $t$, with probability of $p$ there is a new file in the set of popular files, and hence the probability that an uncached file is requested by one of the users is $Kp/N$. As a result, the NDT in time slot $t$ for the genie-aided system can be lower bounded as
\begin{equation}\label{insNDT}
\delta_t(\mu,r) \geq (1-Kp/N)\delta_{\text{off,lb}}(\mu,r)+(Kp/N)\delta_{\text{on,lb}}(\mu,r),
\end{equation}
where $\delta_{\text{off,lb}}(\mu,r)$ is the lower bound on the minimum NDT for offline caching in Lemma \ref{lboffan}, while $\delta_{\text{on,lb}}(\mu,r)$ is a lower bound on the minimum NDT for offline caching in which all files but one can be cached. The lower bound \eqref{insNDT} follows since, in the genie-aided system, with probability $1-Kp/N$ the system is equivalent to the offline caching set-up studied in \cite{Avik2}, while, with probability of $Kp/N$, there is one file that cannot be present in the caches.

To obtain the lower bound $\delta_{\text{on,lb}}(\mu,r)$, we note that the set-up is equivalent to that in \cite{Avik2} with the only difference is that one of the requested files by users is no longer partially cached at ENs. Without loss of generality, we assume that file $F_K$ is requested but it is not partially cached. Revising step (67c) in \cite{Avik2}, we can write
\begin{align}\label{upcut1}
 \textrm{H}(S_{[1:(M-l)]}|F_{[1:l]},F_{[K+1:N]}) & \leq \min \Big ((M-l)(K-l-1)\mu, K-l-1 \Big )L,
\end{align}
which is obtained by using the fact that the constrained entropy of the cached content cannot be larger than the overall size of files $F_j$ with $j \in [l+1,K-1]$. Plugging \eqref{upcut1} into \cite[Eq. (66)]{Avik2} and then taking the limit $L \rightarrow \infty$ and $P \rightarrow \infty$, results in \eqref{LP5}. The rest of proof is as in \cite[Appendix I]{Avik2}.
Using \eqref{insNDT} in the definition of long-term average NDT \eqref{avNDT} concludes the proof.

\section{proof of lemma \ref{up_alpha1}}\label{prolemupa}
\noindent To prove Lemma \ref{up_alpha1}, for any given $\alpha > 1$ and $M \geq 2$, we consider separately small cache regime with $\mu \in [0,1/M]$; intermediate cache regime with $\mu \in [1/M,\alpha/M]$ and the high cache regime with $\mu \in [\alpha/M,1]$.

\noindent $\bullet$ \textbf{Small-cache Regime} ($\mathbf{\mu}$ $\in [0,1/M]$):
Using Lemma \ref{lboffan} a lower bound on the minimum NDT can be obtained as 
\begin{equation}
\begin{aligned}\label{minNDTLP2}
\delta^{*}_{\text{off}} (\mu,r) \geq 1+ \frac{K(1-\mu M)}{Mr}
\end{aligned}
\end{equation}
by considering the constraint \eqref{LP2} with $l=0$ and constraint \eqref{LP3}. Using the offline caching and delivery policy in Sec. \ref{OFF} shown in the top left of Fig. \ref{seroff}, the NDT in the regime of interest is 
\begin{align}\label{lowfront2}
\delta_{\text{off,ach}} \Big (\frac{\mu}{\alpha},r \Big ) &= \Big (\frac{\mu M}{\alpha} \Big ) \delta_{\text{E,Coor}} + \Big (1-\frac{\mu M}{\alpha} \Big ) \big [\delta_{\text{E,C-RAN}} + \delta_{\text{F,C-RAN}} \big ] \nonumber \\ 
& \overset{(a)}{\leq} \frac{(M+K-1)\mu}{\alpha} + \Big ( 1+ \frac{K}{Mr} \Big ) \times \Big (1- \frac{\mu M}{\alpha} \Big ) =1+\frac{(K-1)\mu}{\alpha} + \frac{K}{Mr}\Big (1- \frac{\mu M}{\alpha} \Big ),
\end{align}
where $(a)$ is obtained using \eqref{coor_e} and \eqref{cran_e}.  
From \eqref{minNDTLP2} and \eqref{lowfront2}, we have
\begin{align}\label{UPRatio3}
\delta_{\text{off,ach}} \Big (\frac{\mu}{\alpha},r \Big ) - 2\delta^{*}_{\text{off}} (\mu,r)  
& \overset{(a)}{\leq} \frac{K \mu}{r} \Big (2-\frac{1}{\alpha}\Big)+ \frac{(K-1)\mu}{\alpha}-\frac{K}{Mr}  \overset{(b)}{\leq} \frac{K}{Mr} \Big (1-\frac{1}{\alpha}\Big)+ \frac{(K-1)}{\alpha M}  \nonumber \\  
&\overset{(c)}{\leq} \frac{1}{r} + \frac{1}{\alpha} \Big (1-\frac{1}{r} \Big), 
\end{align}
where $(a)$ is obtained by omitting the first negative term; $(b)$ is obtained by using the fact that $\mu \leq 1/M$ and $(c)$ follows from $M \geq K$.

\noindent $\bullet$ \textbf{ Intermediate cache Regime} ($\mathbf{\mu}$ $\in [1/M,\alpha/M]$):
Using Lemma \ref{lboffan} a lower bound on the minimum NDT can be obtained as $\delta^{*}_{\text{off}} (\mu,r) \geq 1$
by considering the constraint \eqref{LP3}. Using this lower bound and the offline caching and delivery policy in Sec. \ref{OFF} shown in the bottom of Fig. \ref{seroff}, we have
\begin{align}\label{MUPRatio3}
\delta_{\text{off,ach}} \Big (\frac{\mu}{\alpha},r \Big ) - 2\delta^{*}_{\text{off}} (\mu,r) & \leq \Big (\frac{\mu}{\alpha} \Big ) \delta_{\text{E,Coop}} -2 +\Big (1-\frac{\mu}{\alpha} \Big ) [\delta_{\text{E,C-RAN}}+\delta_{\text{F,C-RAN}}]  \overset{(a)}{\leq} \frac{1}{r} 
\end{align}
where $(a)$ is obtained using \eqref{coop_e} and \eqref{cran_e} and also the fact that $M \geq K$.

\noindent $\bullet$  \textbf{Large-cache regime ($\mu \in [\alpha/M,1]$)}:
  In this regime, we have 
\begin{align}\label{UPRatio4}
\dfrac{ \delta_{\text{off,ach}} \big (\mu /\alpha,r \big )}{\delta_{\text{off}}^*(\mu,r)} & \overset{(a)}{\leq} \delta_{\text{off,ach}} \big (\mu /\alpha,r \big )  \overset{(b)}{=} \Big (\frac{\mu M/\alpha-1}{M-1} \Big ) \delta_{\text{E,Coop}}+\Big (\frac{M (1-\mu/\alpha)}{M-1} \Big ) \delta_{\text{E,Coor}} \nonumber \\
&  \overset{(c)}{\leq} \delta_{\text{E,Coor}}  \overset{(d)}{\leq} \frac{M+K-1}{M}   \overset{(e)}{\leq} 2  
\end{align}
where $(a)$ is obtained using the fact that the lower bound $\delta^{*}_{\text{off}} (\mu,r) \geq 1$; $(b)$ is obtained using the offline caching and delivery policy in Sec. \ref{OFF} shown in the top right of Fig. \ref{seroff}; $(c)$ is obtained using the fact that NDT is a decreasing function of $\mu$ and it is maximized by setting $\mu=\alpha/M$; $(d)$ is obtained using \eqref{coor_e} and $(e)$ is obtained using $M \geq K$.
Finally, using \eqref{UPRatio3}-\eqref{UPRatio4} concludes the proof.

\ifCLASSOPTIONcaptionsoff
  \newpage
\fi



%
\bibliographystyle{IEEEtran}
\bibliography{IEEEabrv,IEEEexample}

%




\end{document}